\documentclass[conference]{IEEEtran}

\usepackage{multirow}
\usepackage{subcaption}
\usepackage{graphicx}
\usepackage{amsmath}
\usepackage{amssymb}
\usepackage{amsthm}
\usepackage{booktabs}
\usepackage{algorithm}
\usepackage{algorithmic}

\newtheorem{theorem}{Theorem}
\newtheorem{assumption}{Assumption}
\newtheorem{definition}{Definition}
\newtheorem{proposition}{Proposition}
\newtheorem{corollary}{Corollary}

\begin{document}
%
\title{Public Diffusion Models, Private Images: Key-Controlled Inversion for Conditional Reconstruction}

	

%
\author{\IEEEauthorblockN{Lijunxian Zhang,
Weihai Li\IEEEauthorrefmark{1},
Bin Liu and
Zikai Xu}
\IEEEauthorblockA{School of Cyber Science and Technology,
University of Science and Technology of China,
Hefei, China, 230026\\ Email: \texttt{\{ljxzhang@mail., whli@, flowice@, ustcxzk@mail.\}ustc.edu.cn}}}


\IEEEoverridecommandlockouts
\makeatletter\def\@IEEEpubidpullup{1.5\baselineskip}\makeatother
\IEEEpubid{\parbox{\columnwidth}{
		\IEEEauthorrefmark{1} Corresponding Author
}
\hspace{\columnsep}\makebox[\columnwidth]{}}

\maketitle

\begin{abstract}
Diffusion models are often deployed in settings where model parameters are publicly accessible (e.g., open‑source libraries or released checkpoints). This white‑box scenario creates a serious security risk: any user who obtains an intermediate latent representation can invert the process to recover the original input image. Most prior work on access control for generative models assumes a black‑box model (i.e., parameters are kept secret), typically under an honest‑but‑curious adversary. By contrast, we address the more challenging and realistic white‑box setting where all parameters are public.

We present a key‑controlled inversion framework that turns the inherent error propagation of diffusion models, which exponentially amplifies small perturbations, into a security asset. By injecting key‑dependent noise into the inversion formula, we ensure that only a user with the correct key can reconstruct the original image; any other key yields unrecognizable output.

Theoretically, by leveraging existing error-propagation theory for diffusion models, we prove that the resulting ciphertext distribution is IND-CPA secure and derive that the adversary’s advantage is exponentially small in a tunable security parameter, hence negligible for any probabilistic polynomial‑time (PPT) adversary. Experimentally, we validate these security guarantees across several models and datasets and further demonstrate cross-model robustness, that the injected key noise does not amplify the performance drop caused by model discrepancies.
\end{abstract}


%
\IEEEpeerreviewmaketitle

\section{Introduction}
Diffusion models (DMs) are used in commercial image generation services and open‑source libraries. A user who obtains an intermediate latent representation can invert the process to recover the original input. This creates a security risk when the input is sensitive (e.g., a face or a medical image). Existing protections, such as image watermarking~\cite{fernandez2023stable} or adversarial perturbations~\cite{liu2024stable}, are reactive. They detect or respond after the fact rather than prevent misuse before it. For example, image watermarking can prove ownership, but does not stop an adversary from inverting a latent to recover the original image. This motivates the need for an active access control mechanism that determines who can invert a latent to recover the original input.

The first, studied in the context of privacy‑preserving diffusion models~\cite{lei2025secure}, uses homomorphic encryption or secure multi‑party computation to protect a user's prompt while generating a plaintext latent image. These schemes assume the diffusion model parameters remain black‑box to the user. More importantly, they are designed for text‑to‑image generation and do not address image-to-image reconstruction. The second direction fine‑tunes the VAE decoder~\cite{gai2025pcdiff} so that a secret key is required to obtain a perceptually meaningful image.  Both lines of work rely on keeping the model inaccessible. This assumption fails in open‑source settings or when model checkpoints are publicly released, a common practice in the research community.

In this white‑box setting, where all model parameters are public, an adversary can simply download the model and invert any latent they obtain, rendering black‑box protections ineffective. How can one enable active access control without hiding the parameters?

Our solution rests on two key insights: 

First, there exists a strictly invertible sampling formula (e.g., the O‑BELM sampler~\cite{wang2024belm}) that allows exact reconstruction when the forward and backward paths are aligned. This gives us a correctness guarantee: with the correct key, the inversion process recovers the original input without error. 

Second, diffusion models exhibit exponential error propagation~\cite{li2023error}: a small perturbation injected early in the reverse process accumulates as the chain proceeds. This mirrors the avalanche effect in symmetric cryptography, where a tiny change in the key produces a completely different ciphertext. We turn this traditionally undesirable property into a security asset: by injecting a key‑dependent noise into the noise prediction term, any mismatch (i.e., an incorrect key) is exponentially amplified, making the reconstructed output unrecognizable.

Consequently, we utilize the strictly reversible sampler based on O‑BELM and inject a key‑dependent noise into the noise prediction term during inversion. With the correct key, the noise cancels exactly during decryption, yielding lossless reconstruction. With any other key, the mismatch is exponentially amplified along the chain, producing completely unrecognizable output. 

Thus, by leveraging the error propagation of diffusion models, we prove that the distinguishing advantage of any PPT adversary in the IND‑CPA game is exponentially small in a tunable security parameter, hence negligible. Also, experiments on three datasets confirm that correct keys enable recognizable reconstruction, while wrong keys yield severely distorted outputs. An adaptive adversary cannot gain a non‑negligible advantage even after thousands of queries. Overhead is less than 5\% compared to standard inversion.

\begin{figure}[!t]
    \centering
    \begin{subfigure}{0.9\linewidth}
        \includegraphics[width=1.0\linewidth]{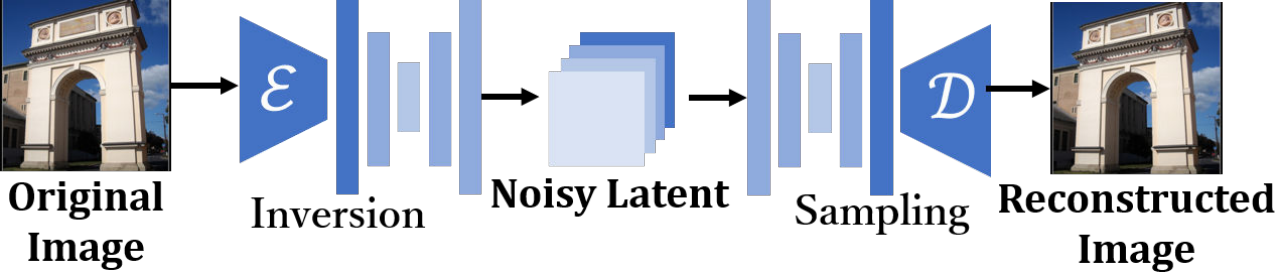}
        \caption{w/o Protection}
    \end{subfigure}
    \begin{subfigure}{0.9\linewidth}
        \includegraphics[width=1.0\linewidth]{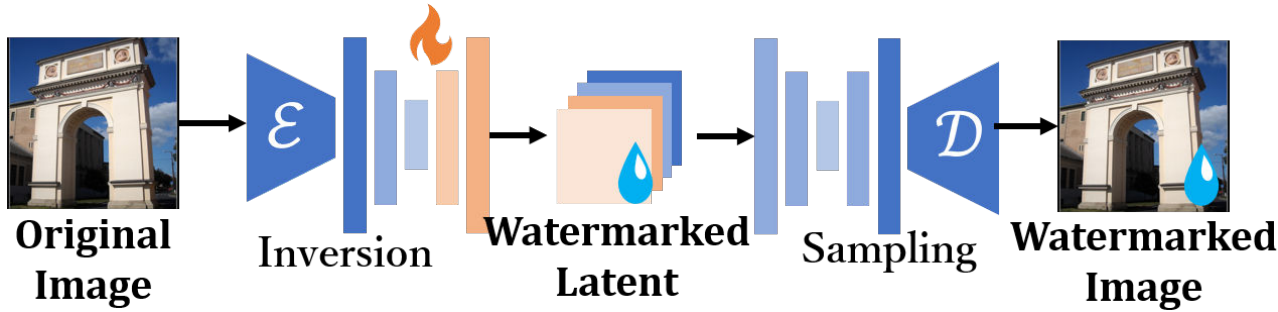}
        \caption{Post-hoc Watermarking~\cite{fernandez2023stable}}
    \end{subfigure}
    \begin{subfigure}{0.9\linewidth}
        \includegraphics[width=1.0\linewidth]{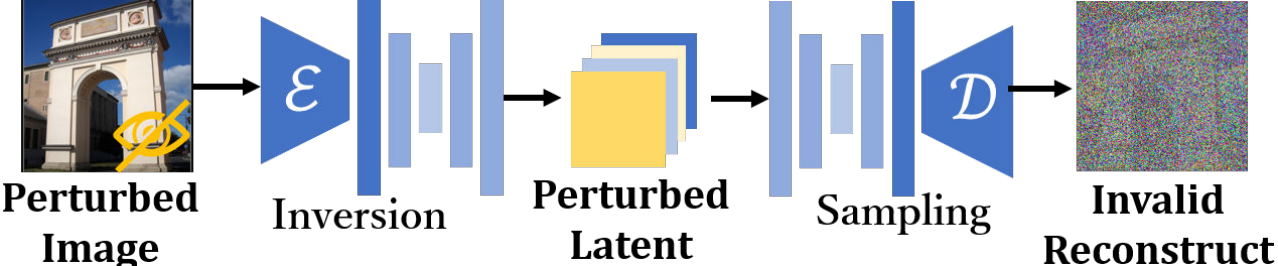}
        \caption{Adversarial Perturbation~\cite{chen2024editshield}}
    \end{subfigure}
    \begin{subfigure}{0.9\linewidth}
        \includegraphics[width=1.0\linewidth]{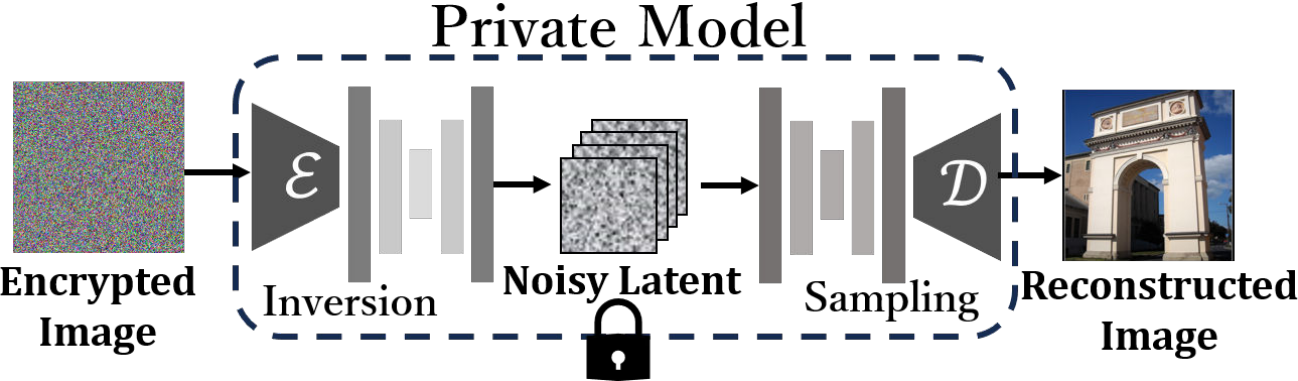}
        \caption{Privacy-Preserving Model~\cite{he2025private}}
    \end{subfigure}
    \begin{subfigure}{0.9\linewidth}
        \includegraphics[width=1.0\linewidth]{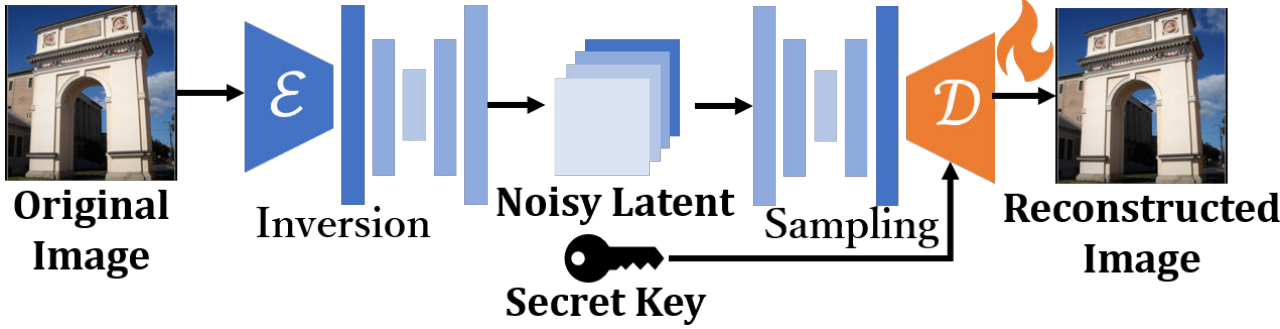}
        \caption{Fine-tuning-based Access Control~\cite{gai2025pcdiff}}
    \end{subfigure}
    \begin{subfigure}{0.9\linewidth}
        \includegraphics[width=1.0\linewidth]{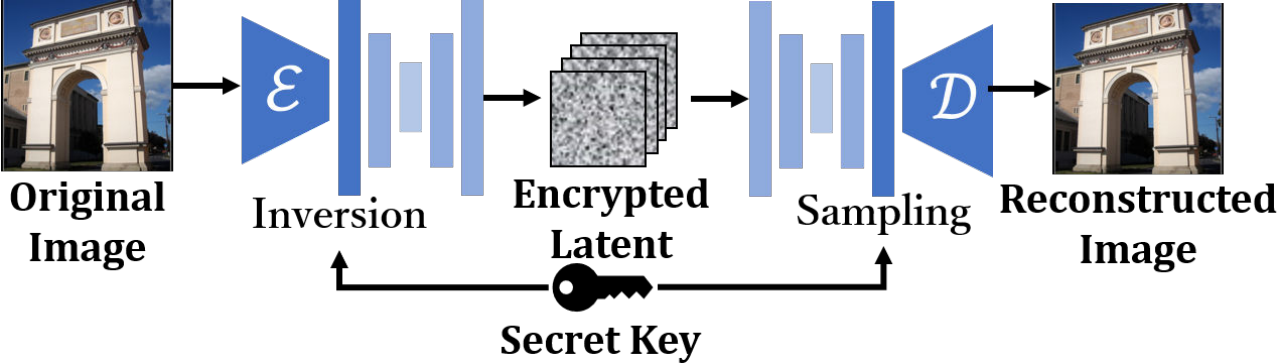}
        \caption{Our Method}
    \end{subfigure}
    \caption{Existing methods for image protection.}
    \label{fig:imageProtect}
\end{figure}

In summary, the main contributions of this work are listed as follows:
\begin{itemize}
    \item Key‑controlled inversion for white‑box diffusion models. We introduce a strictly invertible sampling formula that embeds a secret key into the noise prediction term, enabling active access control without hiding model parameters.
    \item Provable security via error propagation. Leveraging the exponential error amplification property of diffusion models, we conduct a formal security analysis for the key-controlled framework and prove IND‑CPA security with adversary advantage exponentially small in the effective inversion depth.
    \item Comprehensive experimental validation. We evaluate the correctness and security through validation experiments across multiple datasets and models. The results further demonstrate cross-model robustness, i.e., the system remains effective even when encryption and decryption are performed with different models. This property is crucial for practical deployment for frequently updated or fine-tuned models.
\end{itemize}

\section{Background and Preliminaries}
\subsection{Diffusion Models and DDIM Inversion}
Diffusion models are a class of generative models inspired by non-equilibrium thermodynamics~\cite{sohl2015deep}. They operate by gradually destroying the data structure via a forward diffusion process, then learning to reverse it to restore the original data distribution~\cite{jiang2024survey}. Therefore, diffusion models involve two fundamental processes: forward diffusion and reverse denoising. DDPM~\cite{ho2020denoising}formulates them as a probabilistic Markov chain, while DDIM~\cite{song2020denoising} reframes them into an ordinary differential equation (ODE). For DDIM, the forward process is defined as:
\[
\mathbf{x}_t = \sqrt{\bar{\alpha}_t}\,\mathbf{x}_0 + \sqrt{1-\bar{\alpha}_t}\,\epsilon,\quad \epsilon\sim\mathcal{N}(0,\mathbf{I}),
\]
where \(\bar{\alpha}_t = \prod_{i=1}^t \alpha_i\) with \(\alpha_i = 1-\beta_i\) following a predefined schedule. The standard reverse (sampling) step is:
\[
\mathbf{x}_{t-1} = \sqrt{\bar{\alpha}_{t-1}}\,\hat{\mathbf{x}}_{0|t} + \sqrt{1-\bar{\alpha}_{t-1}}\,\epsilon_\theta(\mathbf{x}_t,t),\]
where $\hat{\mathbf{x}}_{0|t}$ satisfies:
\[
\hat{\mathbf{x}}_{0|t} = \frac{1}{\sqrt{\bar{\alpha}_t}}\bigl(\mathbf{x}_t-\sqrt{1-\bar{\alpha}_t}\,\epsilon_\theta(\mathbf{x}_t,t)\bigr).
\]

For image-to-image (I2I) reconstruction, we need the inversion process: given an image \(\mathbf{x}_0\), we compute a sequence of latents \(\mathbf{x}_1,\mathbf{x}_2,\dots,\mathbf{x}_T\) that can later be used to reconstruct \(\mathbf{x}_0\). Unlike the forward process, inversion uses the model’s predicted noise \(\epsilon_\theta(\mathbf{x}_t,t)\) instead of a random sample. The inversion step from \(t\) to \(t+1\) is:
\[
\mathbf{x}_{t+1} = \sqrt{\bar{\alpha}_{t+1}}\,\hat{\mathbf{x}}_{0|t} + \sqrt{1-\bar{\alpha}_{t+1}}\,\epsilon_\theta(\mathbf{x}_t,t),
\]
where \(\hat{\mathbf{x}}_{0|t}\) is defined as above. This expression is obtained by assuming the same noise prediction \(\epsilon_\theta(\mathbf{x}_t,t)\) would have been used in the forward direction. Although deterministic, the forward-reverse pair is not strictly reversible; i.e., applying the reverse step followed by the inversion does not exactly recover the original $\mathbf{x}_0$~\cite{huang2025diffusion}.

\subsection{Strictly Invertible Sampling}
\label{Inv_Re}
The irreversibility stems from a simple asymmetry: the inversion step uses $\epsilon_\theta(\mathbf{x}_t, t)$, while the reverse (sampling) step would need $\epsilon_\theta(\mathbf{x}_{t+1}, t+1)$ to exactly undo it. These two predictions are generally not equal, so the forward-reverse pair is not strictly reversible. Hence, inversion capability is distinct from true reversibility. 

To achieve exact reversible inversion, subsequent works have resorted to fine-tuning~\cite{mokady2023null} or redesigned sampling formulations~\cite{wallace2023edict,zhang2024exact}. In this paper, we adopt O-BELM~\cite{wang2024belm} for diffusion sampling and inversion. It is a bidirectional explicit sampler that guarantees mathematically exact reversibility. Its optimal coefficients, derived from minimizing local truncation error, keep reconstruction error negligible.

Let $\gamma_i$ denote the noise-schedule coefficient used in the O-BELM sampler, defined as $\gamma_i = \sqrt{{\alpha}_i}$, where ${\alpha}_i$ is the standard DDIM parameter above. In the following, we consistently use $\gamma_i$ to avoid confusion with the $\alpha_t$ in DDIM.

The O-BELM sampler and its exact inversion are given by:
\begin{align}
\label{belm_samp}
    \mathbf{x}_{i-1} &= a_i \mathbf{x}_{i+1} + b_i \mathbf{x}_{i} + c_i \,\boldsymbol{\varepsilon}_{\theta}(\mathbf{x}_{i}, i), \\
\label{belm_inv}
    \mathbf{x}_{i+1} &= a_i' \mathbf{x}_{i-1} + b_i' \mathbf{x}_{i} + c_i' \,\boldsymbol{\varepsilon}_{\theta}(\mathbf{x}_{i}, i).
\end{align}
The coefficients $a_i,b_i,c_i$ and $a_i',b_i',c_i'$ are step-dependent functions of $\gamma_i$, derived from minimizing local truncation error:
\begin{align}
    \label{param_val}
    a_i&=\frac{h_{i}^{2}}{h_{i+1}^{2}}\frac{\gamma_{i-1}}{\gamma_{i+1}} \nonumber\\ 
    b_i&=\frac{h_{i+1}^{2}-h_{i}^{2}}{h_{i+1}^{2}}\frac{\gamma_{i-1}}{\gamma_{i}}\\ 
    c_i&=-\frac{h_{i}(h_{i}+h_{i+1})}{h_{i+1}}\gamma_{i-1}\nonumber\\
    \label{param_val2}
    a_i'&=\frac{h_{i+1}^{2}}{h_{i}^{2}}\frac{\gamma_{i+1}}{\gamma_{i-1}}\nonumber\\
    b_i'&=\frac{h_{i}^{2}-h_{i+1}^{2}}{h_{i}^{2}}\frac{\gamma_{i+1}}{\gamma_{i}}\\
    c_i'&=\frac{h_{i+1}(h_{i}+h_{i+1})}{h_{i}}\gamma_{i+1}\nonumber
\end{align}
where $h_i = \bar{\sigma}_i - \bar{\sigma}_{i-1}$, with $\bar{\sigma}_i = \sqrt{(1-\gamma_i)/\gamma_i}$ being the scaled noise level.

These coefficients satisfy the following mutual inverse relations naturally:
\begin{equation}
\label{coef_con}
a_i a_i' = 1,\qquad c_i + a_i' b_i = 0,\qquad c_i' + a_i' c_i = 0.
\end{equation}
Consequently, applying Eq.~(\ref{belm_samp}) followed by Eq.~(\ref{belm_inv}) recovers the original $\mathbf{x}_{i+1}$ exactly, establishing strict reversibility.

\subsection{Error Propagation in Diffusion Models}

\textbf{Diffusion models are sequential}: each denoising step takes the output of the previous step. Consequently, a small error at an early step can propagate and potentially amplify along the chain. This phenomenon, known as error propagation, has been formalized by Li and van der Schaar~\cite{li2023error}. We briefly recall their framework, which we will later use to analyze the security of our key‑controlled inversion.

\begin{definition}[\textbf{Modular and cumulative errors}]
\label{def:mod_cum_err}
     For a diffusion model with $T$ steps, let  $\mathcal{E}_t^{\mathrm{mod}}$ denote the modular error of the $t$-th denoising module $p_\theta(\mathbf{x}_{t-1}\mid\mathbf{x}_t)$ compared with the true denoising process $q(\mathbf{x}_{t-1}\mid\mathbf{x}_t)$:
\begin{equation}
\label{eq:mod_err}
    \mathcal{E}_t^{\mathrm{mod}} = \mathbb{E}_{\mathbf{x}_t\sim p_\theta(\mathbf{x}_t)}\Bigl[D_{\mathrm{KL}}\bigl(p_\theta(\mathbf{x}_{t-1}\mid\mathbf{x}_t)\;\big\|\;q(\mathbf{x}_{t-1}\mid\mathbf{x}_t)\bigr)\Bigr].
\end{equation}

    Similarly, the cumulative error $\mathcal{E}_t^{\mathrm{cum}}$ captures the total discrepancy up to step $t$:
\begin{equation}
    \mathcal{E}_t^{\mathrm{cum}} = D_{\mathrm{KL}}\bigl(p_\theta(\mathbf{x}_{t-1})\;\big\|\;q(\mathbf{x}_{t-1})\bigr),
\end{equation}
where $p_\theta(\mathbf{x}_{t-1})$ is the marginal distribution induced by the sampling chain and $q(\mathbf{x}_{t-1})$ is the true marginal from the forward process. 
\end{definition}

Intuitively, $\mathcal{E}_t^{\mathrm{mod}}$ measures how accurately the module (i.e. $\varepsilon_\theta(\mathbf{x}_t, t)$) predicts the true posterior.  And both errors are non‑negative and $\mathcal{E}_T^{\mathrm{cum}}=0$ because $p_\theta(\mathbf{x}_T)=q(\mathbf{x}_T)$ (standard Gaussian).

Based on these definitions, Li and van der Schaar~\cite{li2023error} established how modular errors accumulate along the chain, leading to the following propagation inequality.

\begin{theorem}[Error propagation inequality \cite{li2023error}]
\label{thm:error_prop}
Under mild conditions (the network output $\epsilon_\theta$ is approximately standard Gaussian and the entropy of $p_\theta(\mathbf{x}_t)$ decreases with $t$), the cumulative error $\mathcal{E}_t^{\mathrm{cum}}$ and modular error $\mathcal{E}_t^{\mathrm{mod}}$ satisfy
\begin{equation}
\label{prog_wo_mu}
\mathcal{E}_t^{\mathrm{cum}} \;\ge\; \mathcal{E}_{t+1}^{\mathrm{cum}} \;+\; \mathcal{E}_t^{\mathrm{mod}}, \qquad \forall t\in[1,T].
\end{equation}
Equivalently, define the amplification factor $\mu_t$ by
\begin{equation}
\label{prog_w_mu}
    \mathcal{E}_t^{\mathrm{cum}} - \mathcal{E}_{t}^{\mathrm{mod}} = \mu_t\,\mathcal{E}_{t+1}^{\mathrm{cum}}.
\end{equation}
Then the inequality implies $\mu_t \ge 1$ for every $t$.
\end{theorem}

\textbf{Empirical Results.} Li and van der Schaar~\cite{li2023error} empirically estimated the cumulative error using the maximum mean discrepancy (MMD)~\cite{gretton2012kernel} with \(T=1000\) sampling steps. Their experiments show that the MMD value remains stable during early steps but grows rapidly as the reverse chain approaches \(t=0\), exhibiting a clear exponential amplification trend. This observation directly supports the existence of a threshold step \(x_0\) after which errors are exponentially magnified, providing an empirical foundation for our Assumption~\ref{asm:linear_growth}.

\section{Threat Model}
\label{thr_mod}
We consider a public diffusion model (white‑box setting) where all parameters are known. A secret key $k$ is used to encrypt an image $\mathbf{x}_0$ into a latent $\mathbf{x}_T^{(\delta)}$ via our key‑controlled inversion. Only a user with the correct key can recover $\mathbf{x}_0$; any other key yields invalid.

\textbf{Adversary Capabilities}: We follow Kerckhoffs's principle: the adversary $\mathcal{A}$ knows the entire diffusion model, the inversion/sampling algorithms, and the encryption scheme, except the secret key $k$. $\mathcal{A}$ is probabilistic polynomial time (PPT) and can adaptively query an encryption oracle $\mathcal{O}$: on input a plaintext image $\mathbf{x}_0$, $\mathcal{O}$ returns $\mathbf{x}_T^{(\delta)}$ generated with a fresh random key noise. $\mathcal{A}$ can make up to $q = \text{poly}(\lambda)$ queries.

\textbf{Adversary Goals}: The adversary $\mathcal{A}$ aims to:
\begin{itemize}
    \item Distinguish which of the two chosen plaintexts corresponds to a given ciphertext (IND‑CPA).
    \item Recover the secret key $k$ or the original image $\mathbf{x}_0$ from a ciphertext.
\end{itemize}

We focus on the first goal, which implies the others.

\textbf{Security Goal: IND‑CPA}: Let $\Pi = (\mathsf{Gen},\mathsf{Enc},\mathsf{Dec})$ be our scheme. For a PPT adversary $\mathcal{A}$, define the IND‑CPA advantage as:
\[
\operatorname{Adv}_{\Pi,\mathcal{A}}^{\mathsf{IND-CPA}}(\lambda) = \left|\Pr[\mathcal{A} \text{ wins}] - \frac12\right|,
\]
where the game proceeds as follows:
\begin{enumerate}
    \item $\mathcal{A}$ is allowed to make up to $q = \mathsf{poly}(\lambda)$ adaptive queries to an encryption oracle $\mathcal{O}$. On each query with a plaintext $\mathbf{x}_0$, $\mathcal{O}$ returns $\mathbf{x}_T^{(\delta)} \gets \mathsf{Enc}(k,\mathbf{x}_0)$ using freshly generated key noise.
    \item $\mathcal{A}$ then chooses two equal‑length plaintexts $\mathbf{x}_0^{(0)}, \mathbf{x}_0^{(1)}$ that were not queried to $\mathcal{O}$.
    \item The challenger picks $b \in \{0,1\}$ uniformly, computes $\mathbf{x}_T^{(\delta)} \gets \mathsf{Enc}(k,\mathbf{x}_0^{(b)})$, and sends it to $\mathcal{A}$.
    \item $\mathcal{A}$ outputs a guess $b'$ and wins if $b' = b$.
\end{enumerate}

$\Pi$ is IND‑CPA secure if for every PPT $\mathcal{A}$, $\operatorname{Adv}_{\Pi,\mathcal{A}}^{\mathsf{IND-CPA}}(\lambda)$ is negligible in the security parameter $\lambda$, where $\lambda$ depends on both the injected noise variance $\sigma_N$ and the effective amplification depth $T$ (the number of steps after which errors grow exponentially). We have proved the existence of such a $\lambda$, and its practical values are empirically given through our experiments.

\section{Proposed Framework}
\subsection{Key-controlled Image Reconstruction}
\begin{figure}[tb]
    \centering
    \includegraphics[width=\linewidth]{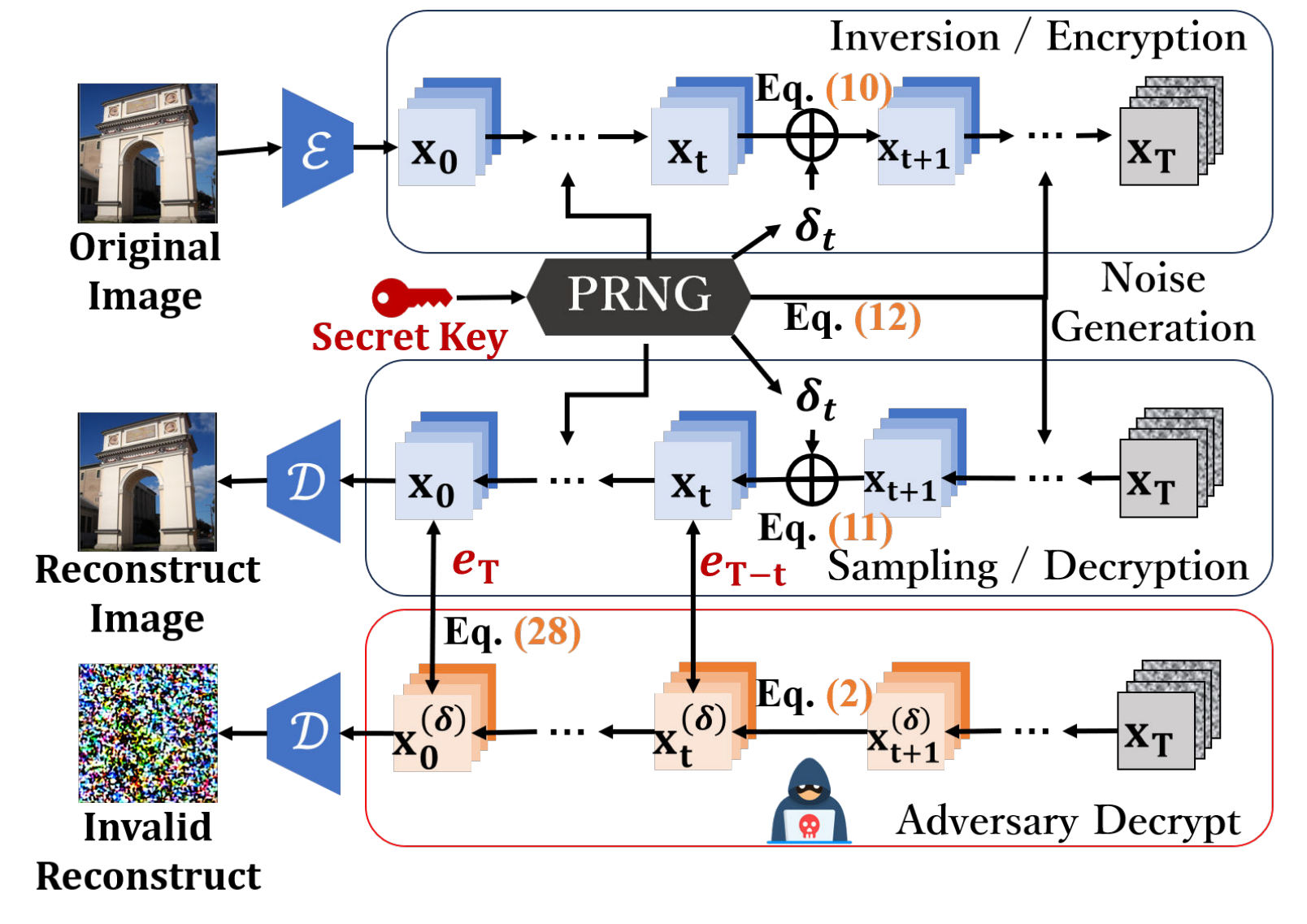}
    \caption{Overall framework. A secret key $k$ drives a PRNG to produce noise $\delta_t$ that is injected into the inversion process. Only the correct key enables exact reconstruction.}
    \label{fig:cryptdiff_framework}
\end{figure}

Our construction starts from the deterministic O‑BELM sampler (Eqs.~\eqref{belm_samp} and~\eqref{belm_inv}), which provides mathematically exact reversibility. To turn this invertible sampler into a key‑controlled encryption scheme, we exploit the fact that the predicted noise term $\boldsymbol{\epsilon}_\theta(\mathbf{x}_i,i)$ propagates through the inversion chain. By augmenting this term with a key‑dependent perturbation $\delta_i$, we obtain a conditional inversion and its corresponding reverse sampling:

\begin{align}
\label{eq:key_samp}
    \mathbf{x}_{i-1} &= a_i \mathbf{x}_{i+1} + b_i \mathbf{x}_i + c_i \bigl(\boldsymbol{\epsilon}_\theta(\mathbf{x}_i,i) + \delta_i\bigr), \\
\label{eq:key_inv}
    \mathbf{x}_{i+1} &= a_i' \mathbf{x}_{i-1} + b_i' \mathbf{x}_i + c_i' \bigl(\boldsymbol{\epsilon}_\theta(\mathbf{x}_i,i) + \delta_i\bigr).
\end{align}

The coefficients $a_i,b_i,c_i$ and $a_i',b_i',c_i'$ satisfy the mutual inverse relations of Eq.~\eqref{coef_con}. Consequently, applying Eq.~\eqref{eq:key_inv} (encryption) followed by Eq.~\eqref{eq:key_samp} (decryption) with the same $\delta_i$ recovers the original $\mathbf{x}_0$ exactly.

\textbf{Key noise generation.}: To make $\delta_i$ reproducible yet unpredictable, we generate it using a cryptographically secure pseudorandom generator (CSPRNG). Let $k$ be a secret key, $n$ a pseudorandom but content-aware seed, and $\sigma_i^2$ a step‑dependent variance. Then
\begin{equation}
    \delta_i = \mathsf{PRNG}(k, i, n) \;\sim\; \mathcal{N}(0, \sigma_i^2 \mathbf{I}_d),
\end{equation}
where $d$ is the dimension of $\mathbf{x}_i$. The variance $\sigma_i^2$ can be a constant or a scheduled function; it trades off security against reconstruction fidelity.

\textbf{Encryption and decryption algorithms.}: Algorithm~\ref{alg:enc} performs encryption (key‑controlled inversion) by iterating Eq.~\eqref{eq:key_inv} from $i=0$ to $T-1$, while Algorithm~\ref{alg:dec} performs decryption (key‑controlled sampling) by iterating Eq.~\eqref{eq:key_samp} backward from $i=T-1$ to $0$.

\begin{algorithm}[tb]
    \caption{Encryption (Key‑controlled Inversion)}
    \label{alg:enc}
    \begin{algorithmic}[1]
        \REQUIRE Original latent $\mathbf{x}_0$, secret key $k$, total steps $T$, seed $n$, variance schedule $\{\sigma_i^2\}_{i=0}^{T-1}$
        \ENSURE Protected latent $\mathbf{x}_{T,\delta}$
        \FOR{$i = 0$ \TO $T-1$}
            \STATE $\delta_i \gets \mathsf{PRNG}(k,i,n) \sim \mathcal{N}(0,\sigma_i^2\mathbf{I}_d)$
            \STATE $\mathbf{x}_{i+1,\delta} \gets a_i' \mathbf{x}_{i-1,\delta} + b_i' \mathbf{x}_{i,\delta} + c_i' (\boldsymbol{\epsilon}_\theta(\mathbf{x}_{i,\delta},i) + \delta_i)$ \COMMENT{Eq.~\eqref{eq:key_inv}}
        \ENDFOR
        \RETURN $\mathbf{x}_{T,\delta}$
    \end{algorithmic}
\end{algorithm}

\begin{algorithm}[tb]
    \caption{Decryption (Key‑controlled Sampling)}
    \label{alg:dec}
    \begin{algorithmic}[1]
        \REQUIRE Protected latent $\mathbf{x}_{T,\delta}$, secret key $k$, total steps $T$, seed $n$, variance schedule $\{\sigma_i^2\}_{i=0}^{T-1}$
        \ENSURE Reconstructed latent $\mathbf{x}_{0,\delta}$
        \FOR{$i = T-1$ \TO $0$}
            \STATE $\delta_i \gets \mathsf{PRNG}(k,i,n) \sim \mathcal{N}(0,\sigma_i^2\mathbf{I}_d)$
            \STATE $\mathbf{x}_{i-1,\delta} \gets a_i \mathbf{x}_{i+1,\delta} + b_i \mathbf{x}_{i,\delta} + c_i (\boldsymbol{\epsilon}_\theta(\mathbf{x}_{i,\delta},i) + \delta_i)$ \COMMENT{Eq.~\eqref{eq:key_samp}}
        \ENDFOR
        \RETURN $\mathbf{x}_{0,\delta}$
    \end{algorithmic}
\end{algorithm}

The next sections analyze the correctness (Proposition~\ref{idel_corr}) and the IND‑CPA security (Theorem~\ref{ind_sec}) of this construction.

\subsection{Correctness Analysis}
The classical correctness criterion requires:
\begin{equation}
    \forall k\in K,m\in M\quad\Pr(\mathsf{Dec}(k,\mathsf{Enc}(k,m))=m)=1
\end{equation}
where $k$ is a key from the key space $K$ and $m$ is a plaintext from the plaintext space $M$. In our scheme, however, perfect reconstruction is relaxed due to the diffusion model’s generative noise $\varepsilon_\theta$. We therefore define correctness as:
\begin{definition}[$\varepsilon$-Correctness]
\label{correct_def}
Our key-controlled inversion scheme is $\varepsilon$-correct if for any secret key $k$ and original latent $\mathbf{x}_0$, the latent $\mathbf{x}_0'$ reconstructed via the conditioned sampling (Alg.~\ref{alg:dec}) from the protected latent $\mathbf{x}_{\delta,T}$ (Alg.~\ref{alg:enc}) satisfies:
\[
\Pr\left( \| \mathbf{x}_0' - \mathbf{x}_0 \| \geq \varepsilon \right) \leq \mathsf{negl}(\lambda),
\]
where the probability encompasses all randomness in the noise injection and model sampling, and $\mathsf{negl}(\cdot)$ is negligible in the security parameter $\lambda$.
\end{definition}

Def.~\ref{correct_def} collapses to perfect (deterministic) correctness as $\varepsilon=0$. With the explicit reversibility of O-BELM, we can establish the following ideal-case guarantee:
\begin{proposition}[Ideal Correctness]
\label{idel_corr}
    Under ideal conditions (no numerical error, deterministic model outputs), our scheme achieves 0-Correctness for any noise variance $\sigma_t$ and coefficient scheduling algorithm.
\end{proposition}

\proof We substitute Eq.~\eqref{eq:key_inv} into Eq.~\eqref{eq:key_samp} under Eq.~\eqref{coef_con}. Then, we can notice that all key-noise terms $\delta$ are cancelled. Hence, $\mathbf{x}_0' = \mathbf{x}_0$ holds, reaching the ideal correctness.

As Prop.~\ref{idel_corr} establishes perfect reconstruction ideal conditions, we can infer that the reconstruction error under the correct key originates only from practical non-idealities such as rounding errors and floating-point calculation errors. This makes $\varepsilon$-correctness an applicable measure in real-world scenarios.

\section{Security Analysis}
\label{sec_proof}

\subsection{Dynamics of Error Propagation}

\textbf{Single-Step KL Divergence Increase}: We analyze how the key noise $\delta_i$ injected at each inversion step propagates through the chain. Before deriving the incremental modular error, we state a standard assumption used in the error propagation literature~\cite{li2023error}.

\begin{assumption}[Distributional Consistency of Noiseless Inversion]
\label{assum:gaussian}
    Applying the inversion process to a plaintext latent $\mathbf{x}_0$ without added key noise yields a final latent $\mathbf{x}_T$ that approximately follows a standard normal distribution: $\mathbf{x}_T \sim \mathcal{N}(\mathbf{0},\mathbf{I})$. 
    
    Moreover, the model's noise prediction $\boldsymbol{\epsilon}_\theta(\mathbf{x}_t,t)$ is trained to fit independent Gaussian noise. Therefore, its output distribution is also approximately $\mathcal{N}(\mathbf{0},\mathbf{I})$ for any input $\mathbf{x}_t$ that is close to the training distribution.
\end{assumption}

This assumption is consistent with the training objective of diffusion models and is widely adopted in theoretical analyses (e.g.,~\cite{li2023error, mokady2023null}). Consequently, the conditional distributions $q(\mathbf{x}_{i+1}\mid\mathbf{x}_i)$ (true posterior) and $p_\theta(\mathbf{x}_{i+1}\mid\mathbf{x}_i)$ (model prediction) can be well approximated by isotropic Gaussians with a known variance $\tilde{\beta}_i$, which is determined by the noise schedule and can be explicitly demonstrated by $\alpha_i$.

Now consider the O‑BELM inversion (encryption) step (Eq.~\eqref{eq:key_inv}), where $\delta_i$ is the key‑dependent noise, independent of $\mathbf{x}_i$. 
Under Assumption~\ref{assum:gaussian}, we have:
\begin{align}
    q(\mathbf{x}_{i+1}\mid\mathbf{x}_i) &= \mathcal{N}\bigl(\boldsymbol{\mu}_i^{\text{true}},\; \tilde{\beta}_i \mathbf{I}_d\bigr),\\
    p_\theta(\mathbf{x}_{i+1}\mid\mathbf{x}_i) &= \mathcal{N}\bigl(\boldsymbol{\mu}_i,\; \tilde{\beta}_i \mathbf{I}_d\bigr),
\end{align}
where $d$ is the dimension of $\mathbf{x}_i$, and
\begin{equation}
    \boldsymbol{\mu}_i = \mathbb{E}[a_i' \mathbf{x}_{i-1} + b_i' \mathbf{x}_i + c_i' \boldsymbol{\epsilon}_\theta(\mathbf{x}_i,i)].
\end{equation}
Injecting $\delta_i$ shifts the mean to $\boldsymbol{\mu}_i' = \boldsymbol{\mu}_i + c_i'\delta_i$.

For two Gaussians with identical isotropic covariance, the KL divergence is
\begin{equation}
    D_{\mathrm{KL}}\bigl(\mathcal{N}(\boldsymbol{\mu}_1,\sigma^2\mathbf{I}) \,\|\, \mathcal{N}(\boldsymbol{\mu}_2,\sigma^2\mathbf{I})\bigr) = \frac{\|\boldsymbol{\mu}_1-\boldsymbol{\mu}_2\|^2}{2\sigma^2}.
\end{equation}
Thus, the increase in KL divergence caused by $\delta_i$ is
\begin{equation}
    \Delta D_{\mathrm{KL}}(\mathbf{x}_i,\delta_i) = \frac{c_i'^2}{2\tilde{\beta}_i}\|\delta_i\|^2 + \frac{c_i'}{\tilde{\beta}_i} (\boldsymbol{\mu}_i - \boldsymbol{\mu}_i^{\text{true}})^\top \delta_i.
\end{equation}

Taking expectation over independent $\mathbf{x}_i\sim p_\theta(\mathbf{x}_i)$ and $\delta_i\sim\mathcal{N}(0,\sigma_{\delta,i}^2\mathbf{I}_d)$ eliminates the cross term ($\mathbb{E}[\delta_i]=0$), yielding
\begin{equation}
\label{eq:dkl}
    {\Delta_i := \mathbb{E}[\Delta D_{\mathrm{KL}}] = \frac{c_i'^2}{2\tilde{\beta}_i} \cdot d \cdot \sigma_{\delta,i}^2}.
\end{equation}
This quantity is the expected increase in modular error contributed by the key noise $\delta_i$ in a single inversion step. It is exact under the Gaussian assumption and depends only on the O‑BELM coefficient $c_i'$, the true posterior variance $\tilde{\beta}_i$, the latent dimension $d$, and the noise variance $\sigma_{\delta, i}^2$.

\textbf{Error Propagation Across Steps}: Recall from Def.~\ref{def:mod_cum_err} that the modular error of step \(i\) is the expected KL divergence between the model's conditional distribution and the true posterior. 

The previous part derived this expectation under the key noise \(\delta_i\); that value is exactly \(\Delta_i\) as defined in Eq.~\eqref{eq:dkl}. Hence, after injecting key noise, the modular error becomes \(\Delta_i\).

Now we invoke the error propagation inequality of Li and van der Schaar (Thm.~\ref{thm:error_prop}), which holds for the true nonlinear diffusion model:
\[
\mathcal{E}_t^{\mathrm{cum}} \;\ge\; \mathcal{E}_{t+1}^{\mathrm{cum}} \;+\; \mathcal{E}_t^{\mathrm{mod}}, \qquad \forall t\in[1,T],
\]
where \(\mathcal{E}_t^{\mathrm{cum}}\) is the cumulative error at step \(t\) and \(\mathcal{E}_t^{\mathrm{mod}}\) is the modular error. 
When key noise is injected, the modular error \(\mathcal{E}_i^{\mathrm{mod}}\) is replaced by \(\Delta_i\).  
Iterating the inequality from \(t=T-1\) down to \(0\) and assuming \(\mathcal{E}_T^{\mathrm{cum}}=0\) (because the backward chain starts from pure Gaussian noise), we obtain the following lower bound on the final cumulative error.

\begin{proposition}[Cumulative error lower bound under key noise]
\label{prop:cumulative_bound}
Under the conditions of Thm.~\ref{thm:error_prop}, injecting key noise that increases the modular error at step \(i\) to \(\Delta_i\) yields
\[
\mathcal{E}_0^{\mathrm{cum}} \;=\; \sum_{i=0}^{T-1} \Bigl( \prod_{s=0}^{i} \mu_s \Bigr) \Delta_i,
\]
where \(\mu_s \ge 1\) are the amplification factors.
\end{proposition}
\begin{proof}
We start from the inequality \(\mathcal{E}_t^{\mathrm{cum}} \ge \mathcal{E}_{t+1}^{\mathrm{cum}} + \mathcal{E}_t^{\mathrm{mod}}\) and rewrite it using the amplification factor \(\mu_t\) defined by
\[
\mathcal{E}_t^{\mathrm{cum}} - \mathcal{E}_t^{\mathrm{mod}} = \mu_t \mathcal{E}_{t+1}^{\mathrm{cum}}.
\]
Thm.~\ref{thm:error_prop} guarantees \(\mu_t \ge 1\) for all \(t\). Substituting the modular error \(\mathcal{E}_t^{\mathrm{mod}}\) with the key‑noise induced value \(\Delta_t\) (and noting that the original modular error is zero in the ideal model, or that \(\Delta_t\) already accounts for the increment), we obtain
\[
\mathcal{E}_t^{\mathrm{cum}} = \mu_t \mathcal{E}_{t+1}^{\mathrm{cum}} + \Delta_t.
\]
Iterating this relation from \(t=T-1\) down to \(0\) with \(\mathcal{E}_T^{\mathrm{cum}}=0\) gives
\[
\mathcal{E}_0^{\mathrm{cum}} = \sum_{i=0}^{T-1} \Bigl( \prod_{s=0}^{i} \mu_s \Bigr) \Delta_i.
\]
Since each \(\mu_s \ge 1\), the right‑hand side is at least the sum of the \(\Delta_i\). This establishes the bound.
\end{proof}

This proposition shows that even small per‑step perturbations \(\Delta_i\) lead to an exponentially huge final cumulative error when the key is incorrect, as long as the amplification factors satisfy \(\mu_t > 1\) for most steps. Consequently, under such conditions, the ciphertext distributions corresponding to different plaintexts become statistically indistinguishable.

\subsection{Linearized Differential Analysis}
Prop.~\ref{prop:cumulative_bound} shows that the cumulative error grows at least as \(\prod_{s=0}^{t}\mu_s\) when key noise is injected. However, to obtain a quantitative bound on the adversary's IND‑CPA advantage, we need a more explicit expression for the covariance of the ciphertext distribution. For this purpose, we consider a differential analysis for the linearized approximation of the error dynamics.

Given that each key noise $\delta_t$ is i.i.d. from $\mathcal{N}(0,\sigma_t^2\mathbf{I}_d)$, we can analyze the general attack scenario by considering a null noise estimate ($\delta = 0$). Leveraging the reversibility of O-BELM, the error propagation is therefore traced through the encryption process.

Considering the unconditional inversion equation Eq.~\eqref{belm_inv} and the key-conditioned version Eq.~\eqref{eq:key_inv}, let $\mathbf{e}_i$ denote the latent error at step $i$, i.e., $\mathbf{e}_i = \mathbf{x}_{i,\delta}-\mathbf{x}_i$. Subtracting these two equations gives the dynamical equation for error propagation:
\begin{equation}
    \label{error_prop}
    \mathbf{e}_{i+1}=a_i'\mathbf{e}_{i-1}+b_i'\mathbf{e}_i+c_i'[\boldsymbol{\varepsilon}_{\theta}(\mathbf{x}_{i,\delta},i) - \boldsymbol{\varepsilon}_{\theta}(\mathbf{x}_{i},i)] +c_i'\delta
\end{equation}
We assume $\varepsilon_\theta(\mathbf{x},t)$ is continuously differentiable in $\mathbf{x}$. The property is achievable via differentiable activations. This allows us to simplify Eq.~\eqref{error_prop} through a first-order approximation:
\begin{equation}
    \label{lin_error_prop}
    \mathbf{e}_{i+1}=a_i'\mathbf{e}_{i-1}+(b_i'\mathbf{I} + c_i'\mathbf{E}_i)\mathbf{e}_i +c_i'\delta
\end{equation}
where $\mathbf{E}_i$ is the Jacobian matrix of $\boldsymbol{\varepsilon}_{\theta}$, given by:
\begin{equation}
    \label{1order_approx}
    \boldsymbol{\varepsilon}_{\theta}(\mathbf{x}_{i,\delta},i) - \boldsymbol{\varepsilon}_{\theta}(\mathbf{x}_{i},i) \approx \mathbf{E}_i(\mathbf{x}_{i,\delta} - \mathbf{x}_{i}) = \mathbf{E}_i\mathbf{e}_i.
\end{equation}

This forms a second-order linear system governed by Eq.~\eqref{lin_error_prop}. To analyze its signal amplification, we rewrite it in vector form as shown below.
\begin{equation}
    \label{vec_error_prop}
    \begin{bmatrix}
        \mathbf{e}_{i+1}\\
        \mathbf{e}_{i}
    \end{bmatrix}=\begin{bmatrix}
        b_i'\mathbf{I} + c_i'\mathbf{E}_i & a_i'\mathbf{I}\\
       \mathbf{I} & \mathbf{0}
    \end{bmatrix}\begin{bmatrix}
        \mathbf{e}_{i}\\
        \mathbf{e}_{i-1}
    \end{bmatrix} +c_i'\begin{bmatrix}\delta\\\mathbf{0}\end{bmatrix}
\end{equation}
Eq.~\eqref{vec_error_prop} can be decomposed into its homogeneous part and driven part. The homogeneous part is governed by the state-transition matrix $\mathbf{M}_i$.

\textbf{Connection to nonlinear amplification.} Noting the relationship between the linearized approximation and the true nonlinear model, we have the following proposition.

\begin{proposition}[]
\label{prop:spectral_lower}
Let $\mathbf{M}_i$ be the state transition matrix of the linearized error dynamics (Eq.~\eqref{vec_error_prop}) and let $\mu_i$ be the amplification factor of the true nonlinear model (Theorem~\ref{thm:error_prop}). Under the same conditions as Theorem~\ref{thm:error_prop} (Li \& van der Schaar), the spectral radius satisfies $\rho(\mathbf{M}_i) \ge 1$ for every step $i$.
\end{proposition}

\begin{proof}
Theorem~\ref{thm:error_prop} guarantees $\mu_i \ge 1$ for every step. This means that any error injected at step $i$ cannot be reduced by the subsequent dynamics; it is at least preserved when propagated to the next step.

Now suppose, for contradiction, that $\rho(\mathbf{M}_i) < 1$. Then, for sufficiently small initial errors, the linearized system would predict exponential contraction over successive steps. Since the true nonlinear dynamics is a first‑order perturbation of the linearized system (the difference consists of higher‑order terms in the error), the actual evolution would also contract. This directly contradicts the fact that $\mu_i \ge 1$ (which holds for arbitrarily small errors). Hence our assumption $\rho(\mathbf{M}_i) < 1$ is impossible, and we must have $\rho(\mathbf{M}_i) \ge 1$.
\end{proof}

Prop.~\ref{prop:spectral_lower} shows that the linearized approximation is consistent with the true nonlinear dynamics, ensuring that the tractable error bounds obtained from the linear model are meaningful and provide a valid basis for security analysis.

\textbf{Explicit form of linearized errors.} Now we solve the linear recurrence. The total error $\mathbf{e}_T$ after $T$ iterations can be expressed as a closed‑form sum of the noises injected at each step, each amplified by the propagation gain $\mathbf{\Psi}_t$:
\begin{align}
\label{total_error}
    \mathbf{e}_{T} & =\sum_{t=0}^{T-1}\mathbf{\Psi}_t\delta_t, \\ 
\label{explicit_psi}
    \mathbf{\Psi}_t &= \mathbf{P}\,\mathbf{\Phi}(T,t)\,\mathbf{N}_t
\end{align}
where $\mathbf{\Phi}(T,t) = \prod_{i=t}^{T-1}\mathbf{M}_i$ is the cumulative state transition from step $t$ to $T$, $\mathbf{P}=[\mathbf{0}\quad \mathbf{I}]$ extracts the error component from the full state vector, and $\mathbf{N}_t=[\mathbf{0}\quad c_t'\mathbf{I}]^T$ maps the injected noise $\delta_t$ into the state space. Thus, $\mathbf{\Psi}_t$ quantifies how the noise injected at step $t$ contributes to the final error after forward propagation.

Substituting Eq.~\eqref{explicit_psi} into Eq.~\eqref{total_error}, we directly obtain the following bound on the overall error magnitude:

\begin{align}
    \|\mathbf{e}_T\| & = \left\| \sum_{t=0}^{T-1} \mathbf{P}\,\mathbf{\Phi}(T,t)\,\mathbf{N}_t \delta_t \right\| \nonumber \\
                          & \leq \kappa \cdot \max_t \|\delta_t\| \cdot \sum_{t=0}^{T-1} \|\mathbf{\Phi}(T,t)\|,
    \label{error_bound_phi}
\end{align}
where $\kappa = \|\mathbf{P}\| \max_t \|\mathbf{N}_t\|$ is a constant determined by coefficient schedule.
This leads to a useful characterization of error growth:
\begin{corollary}[Error Propagation Bound] \label{cor:error_prop_bound}
    In the linearized dynamics, the norm of the total latent error $\|\mathbf{e}_T\|$ satisfies:
    \begin{equation}
        \|\mathbf{e}_T\| \leq C \cdot \left( \sum_{t=0}^{T-1} \|\mathbf{\Phi}(T,t)\| \right)
    \end{equation}
    where $C = \kappa \cdot \max_t \|\delta_t\|$ is a constant factor. The product of the maximum injected noise magnitude and the sum of the propagation-gain norms $\|\mathbf{\Psi}_t\|$. Consequently, the error amplification over the entire path is governed by the sum of $\|\mathbf{\Phi}(T,t)\|$.
\end{corollary}

\subsection{Indistinguishability of Ciphertext Distributions}
In the standard IND-CPA game, an adversary $\mathcal{A}$ with oracle access to $\mathsf{Enc}(k,\cdot)$ submits two plaintexts $(m_0, m_1)$ and receives a challenge ciphertext $c^* \leftarrow \mathsf{Enc}(k, m_b)$ for a randomly chosen $b \in \{0,1\}$. The adversary's advantage is defined as $\mathsf{Adv}^{\mathrm{IND-CPA}}_{\mathcal{A},\Pi} = \bigl| \Pr[b' = b] - \frac{1}{2} \bigr|$, where $b'$ is $\mathcal{A}$'s guess. The scheme is IND-CPA secure if $\mathsf{Adv}^{\mathrm{IND-CPA}}_{\mathcal{A},\Pi}$ is negligible for any probabilistic polynomial-time $\mathcal{A}$.

Since the challenge ciphertext $c^*$ is the protected latent $\mathbf{x}^{(b)}_{T,\delta}$, the adversary’s distinguishing task reduces to a statistical hypothesis test (Def.~\ref{Hyp_test}). To validate this claim, we analyze the character of the ciphertext distribution. For any plaintext $m$ corresponding to the original latent $\mathbf{x}_0$, we use Eq.~\eqref{total_error} and \eqref{explicit_psi} to relate the ciphertext $\mathbf{x}_{T,\delta}$ and the noiseless inversion result $\mathbf{x}_{T,0}$ as follows.
\begin{align}
    \mathbf{x}_{T,\delta} & = \mathbf{x}_{T,0} + \mathbf{e}_T \nonumber\\
    & = \mathbf{x}_{T,0} + \sum_{t=0}^{T-1} \mathbf{P}\mathbf{\Phi}(T,t)\mathbf{N}_t\delta_t
    \label{xt_from_x0}
\end{align}
Moreover, as we assume that each $\delta_t$ is i.i.d. sampled from $\mathcal{N}(0,\sigma_t^2\mathbf{I}_d)$, the distribution of the ciphertext can be closely approximated by:

\begin{equation}
    \label{ciph_dis}
    \mathbf{x}_{T,\delta} \sim \mathcal{N}(\mathbf{x}_{T,0},\mathbf{\Sigma}_T)
\end{equation}
where $\mathbf{\Sigma}_T$ satisfies:
\begin{equation}
    \label{sigma_sat}
    \mathbf{\Sigma}_T = \sum_{t=0}^{T-1}\sigma_t^2\mathbf{\Psi}_t^T\mathbf{\Psi}_t \leq (\max_t\sigma_t)^2\sum_{t=0}^{T-1}\mathbf{\Psi}_t^T\mathbf{\Psi}_t.
\end{equation}

Therefore, given the ciphertext distribution, we define the hypothesis testing problem as follows:

\begin{definition}[Hypothesis Test for IND-CPA Distinguishing]
\label{Hyp_test}
Let $\mathcal{P}_0 = \mathcal{N}(\mathbf{x}_{T,0}^{(0)}, \boldsymbol{\Sigma}_T)$ and $\mathcal{P}_1 = \mathcal{N}(\mathbf{x}_{T,0}^{(1)}, \boldsymbol{\Sigma}_T)$ denote the two possible ciphertext distributions corresponding to the encryption of messages $\mathbf{x}^{(0)}_0$ and $\mathbf{x}^{(1)}_0$, respectively. Given a single observed ciphertext $\mathbf{x}_{T,\delta}^{(b)}$, $b \in \{0,1\}$, the adversary's task is to decide between:
\begin{align}
H_0 &: \mathbf{x}_{T,\delta}^{(b)} \sim \mathcal{P}_0, \\
H_1 &: \mathbf{x}_{T,\delta}^{(b)} \sim \mathcal{P}_1.
\end{align}
This constitutes a hypothesis test for two multivariate Gaussian distributions with a common covariance matrix $\boldsymbol{\Sigma}_T$.
\end{definition}

According to Def.~\ref{Hyp_test}, we can find that the minimum error probability is:
\begin{align}
    \label{min_err}
    P_{err}^*&=\Phi(-\frac{d_{\text{CT}}}{2})\qquad\Phi(x) = \frac{1}{\sqrt{2\pi}}\int_{-\infty}^x e^{t^2/2}\text{d}t\\
    \label{maha_dist}
    d_{\text{CT}} &= \sqrt{(\mathbf{x}_{T,0}^{(0)}-\mathbf{x}_{T,0}^{(1)})^T\mathbf{\Sigma}_T^{-1}(\mathbf{x}_{T,0}^{(0)}-\mathbf{x}_{T,0}^{(1)})}
\end{align}
where $\Phi(\cdot)$ denotes the cumulative distribution function (CDF) of the standard normal distribution, and $d_{\text{CT}}$ is the Mahalanobis distance between the two distributions. Consequently, under this setting, the adversary’s advantage is upper-bounded by:
\begin{equation}
    \label{adv_cpa}
    \mathsf{Adv}^{\mathrm{IND-CPA}}_{\mathcal{A},\Pi} = \left|\frac{1}{2} - P_{\text{err}}^*\right|.
\end{equation}
The diffusion model's noiseless inversion approximates a high-dimensional normal sample, so the adversary's advantage is essentially bounded by the covariance $\boldsymbol{\Sigma}_T$. 

To quantify this bound, we introduce two related assumptions. The first captures the essential nonlinear error propagation property, while the second translates it into a form suitable for linearized analysis.

\begin{assumption}[Exponential Amplification in the Nonlinear Model]
\label{asm:nonlinear_amplification}
There exists a threshold step $x_0 < N$ such that for all steps $t > x_0$, the amplification factor satisfies $\mu_t > 1$. Moreover, the product $\prod_{s=x_0+1}^{t} \mu_s$ grows exponentially. 
\end{assumption}

This property is consistent with the empirical results of~\cite{li2023error} and ensures that errors injected at earlier steps are exponentially magnified. The following assumption is a linearized counterpart, which is justified by Proposition~\ref{prop:spectral_lower}.

\begin{assumption}[Exponential Growth in the Linearized Model]
\label{asm:linear_growth}
For a fixed maximum inversion step $N$, there exists a threshold $x_0 < N$ such that for any $x > x_0$ and $t > 0$,
\[
\|\boldsymbol{\Phi}(N, x+t)\| \sim \rho^{\,t} \|\boldsymbol{\Phi}(N, x)\|,
\]
with $\rho > 1$, while for $1 \le x \le x_0$, $\|\boldsymbol{\Phi}(N, x)\| \approx k \|\boldsymbol{\Phi}(N, 1)\|$ for some constant $k$. 
\end{assumption}

This exactly matches the error propagation pattern observed in~\cite{li2023error} when the key-induced noise is treated as part of the modular error. Based on these assumptions, we now state the following propositions. Their proof is given in the Appendix.

\begin{theorem}
\label{secure_samp}
Under Asm.~\ref{asm:linear_growth}, there exists a noise schedule such that distinguishing the ciphertext distributions corresponding to two distinct plaintexts is computationally hard. More precisely, the distinguishing advantage is exponentially small in the number of effective steps \(N - x_0\).
\end{theorem}

\begin{theorem}
\label{ind_sec}
If Thm.~\ref{secure_samp} holds and the PRNG is cryptographically secure (pseudorandom), then there exists a sampling configuration for which the key-controlled inversion scheme is IND-CPA secure. And the adversary's advantage is bounded by \(\exp(-\Omega(\lambda))\), which is negligible for any PPT adversary.
\end{theorem}

\subsection{Security Boundaries and Parameters}

We now examine the practical security guarantees of our scheme, returning to the threat model defined in Section~\ref{thr_mod}. 

\textbf{Multi-query indistinguishability.} Our IND-CPA analysis previously focused on a single challenge. However, the adversary in the IND-CPA game is allowed to make up to \(q = \mathsf{poly}(\lambda)\) adaptive queries before receiving the challenge ciphertext.  

Because each encryption uses fresh, independently generated pseudo‑random noise \(\{\delta_t\}\) (derived from a secret key \(k\) and a seed \(n\) to prevent replay), the ciphertexts from different queries are independent.

This independence enables a standard hybrid argument. Define hybrid experiments \(\mathsf{Hyb}_0, \mathsf{Hyb}_1, \dots, \mathsf{Hyb}_q\) as follows: in \(\mathsf{Hyb}_i\), the first \(i\) queries are answered with the true encryption of the adversary’s chosen plaintexts, while the remaining \(q-i\) queries are answered with encryptions of a fixed reference plaintext (e.g., the all‑zero image).

The difference between \(\mathsf{Hyb}_{i-1}\) and \(\mathsf{Hyb}_i\) lies exactly in the \(i\)-th query, which reduces to the single‑challenge distinguishing scenario. Hence, the advantage of distinguishing \(\mathsf{Hyb}_{i-1}\) from \(\mathsf{Hyb}_i\) is at most the single‑challenge advantage \(\varepsilon = \exp(-\Omega(\lambda))\).

By the triangle inequality, the total advantage in distinguishing the real game (\(\mathsf{Hyb}_q\)) from the all‑reference game (\(\mathsf{Hyb}_0\)) is bounded by \(q \cdot \varepsilon\). Since \(q = \mathsf{poly}(\lambda)\) and \(\varepsilon\) is exponentially small, the overall advantage remains negligible. Thus our single‑challenge security analysis directly extends to the full multi‑query IND‑CPA setting.

\textbf{Key \& Plaintext Recovery.} The key noise is generated by a CSPRNG; extracting the key from a ciphertext is as hard as breaking the CSPRNG. Plaintext recovery is implied by IND‑CPA security: a successful plaintext recovery would directly yield a distinguishing advantage. Thus, both attacks are computationally infeasible under the same assumptions.

\textbf{Practical determination of the security parameter.} The theoretical analysis establishes that the adversary's advantage is bounded by an exponentially decaying function of $\lambda = N - x_0$, where $x_0$ is the step after which the amplification factors $\mu_t$ become strictly greater than $1$.

In practice, however, the exact value of $x_0$ depends on the noise schedule, the model architecture, and the statistics of the training data. Moreover, the total number of steps $N$ cannot be arbitrarily large because the diffusion model's own prediction error accumulates over long chains, potentially deteriorating the reconstruction quality for legitimate users.

Therefore, while the existence of a security parameter $\lambda$ is guaranteed, its concrete numerical value must be determined empirically. We provide a set of experiments in Sec.~\ref{sec_eva} that measure the reconstruction error under correct and incorrect keys for various $N$ and noise variances, thereby offering practical guidance for choosing $\lambda$ to achieve a desired security level.

\section{Experiments}
\subsection{Experimental Setup}
\textbf{Datasets and Models}. We conduct experiments on three standard image datasets: ImageNet-1K~\cite{deng2009imagenet}, MS-COCO~\cite{lin2014microsoft}, and CelebA-HQ~\cite{karras2017progressive}. We randomly select 200 images with the size of $512 \times 512$ from each dataset. Moreover, we employ multiple public latent diffusion models: Stable Diffusion~\cite{rombach2021high} v1.4, v1.5, v2.0, and v2.1 to evaluate cross-model robustness. 

\textbf{Environment and Parameters}. All experiments are performed on an RTX 4090 GPU with PyTorch 2.7. We use the inversion and sampling process of the O-BELM sampler for encryption and decryption, respectively. The default number of steps is set to 20 for both processes, and the noise schedule follows the squared cosine (\texttt{squaredcos\_cap\_v2}) scheduler~\cite{nichol2021improved}. 

\textbf{Key Noise Generation}. The key noise \(\delta_t\) is sampled from a Gaussian distribution \(\mathcal{N}(0, \sigma_{\delta}^2 \mathbf{I}_d)\). 
To meet the fixed input length requirement of the PRNG and to cryptographically bind the noise to the secret key, the step index, and the query, we first compute a hash:
\[
\text{seed}_t = \mathsf{Hash}(k \,\|\, t \,\|\, n),
\]
where \(k\) is the secret key, \(t\) the step index, \(n\) a query‑specific seed (to prevent replay attacks), and \(\mathsf{Hash}\) is a cryptographic hash function (e.g., SHA‑256). 
This seed is then used to instantiate a cryptographically secure PRNG (AES‑128 in CTR mode) to generate \(\delta_t \sim \mathcal{N}(0, \sigma_{\delta}^2 \mathbf{I}_d)\). 
Unless stated otherwise, we set \(\sigma_{\delta} = 1\) and vary it in Sec.~\ref{Qu_ConRe} to study the effect on legitimate reconstruction.

\textbf{Evaluation Metrics}. To assess reconstruction quality, we adopt three widely used metrics:

\begin{itemize}
    \item \textbf{Peak Signal-to-Noise Ratio (PSNR)} measures pixel‑wise fidelity. For an original image $\mathbf{x}_0$ and its reconstructed version $\hat{\mathbf{x}}_0$, both with values in $[0,1]$, we define
    \[
    \mathrm{PSNR} = 10 \log_{10}\!\left(\frac{1}{\mathrm{MSE}}\right),\quad 
    \mathrm{MSE} = \frac{1}{d}\|\mathbf{x}_0 - \hat{\mathbf{x}}_0\|^2,
    \]
    where $d$ is the number of pixels. Higher PSNR indicates better reconstruction.

    \item \textbf{Structural Similarity Index (SSIM)} evaluates perceptual similarity. It combines luminance, contrast, and structure terms:
    \[
    \mathrm{SSIM} = \frac{(2\mu_{\mathbf{x}}\mu_{\hat{\mathbf{x}}} + C_1)(2\sigma_{\mathbf{x}\hat{\mathbf{x}}} + C_2)}{(\mu_{\mathbf{x}}^2 + \mu_{\hat{\mathbf{x}}}^2 + C_1)(\sigma_{\mathbf{x}}^2 + \sigma_{\hat{\mathbf{x}}}^2 + C_2)},
    \]
    where $\mu$, $\sigma$, $\sigma_{\mathbf{x}\hat{\mathbf{x}}}$ denote local means, standard deviations, and cross‑correlation, and $C_1, C_2$ are small constants. SSIM ranges in $[0,1]$, with 1 meaning identical.

    \item \textbf{Fréchet Inception Distance (FID)}~\cite{heusel2017gans} quantifies the distributional distance between real and reconstructed images. FID is computed as
    \[
    \mathrm{FID} = \|\boldsymbol{\mu}_r - \boldsymbol{\mu}_g\|^2 + \mathrm{Tr}\!\left(\boldsymbol{\Sigma}_r + \boldsymbol{\Sigma}_g - 2(\boldsymbol{\Sigma}_r\boldsymbol{\Sigma}_g)^{1/2}\right),
    \]
    where $(\boldsymbol{\mu}_r,\boldsymbol{\Sigma}_r)$ and $(\boldsymbol{\mu}_g,\boldsymbol{\Sigma}_g)$ are the mean and covariance of features extracted from real and generated images by a pre‑trained Inception network. Lower FID indicates higher fidelity and diversity.
\end{itemize}

We also report computational efficiency as the total runtime for a complete encryption‑decryption cycle. Additional security‑specific metrics (e.g., distinguishing accuracy under adaptive queries) are presented in Sec.~\ref{sec_eva}.
\subsection{Quality of Conditional Reconstruction}
\label{Qu_ConRe}

\begin{figure}
    \centering
    \begin{subfigure}{0.49\linewidth}
        \includegraphics[width=1.0\linewidth]{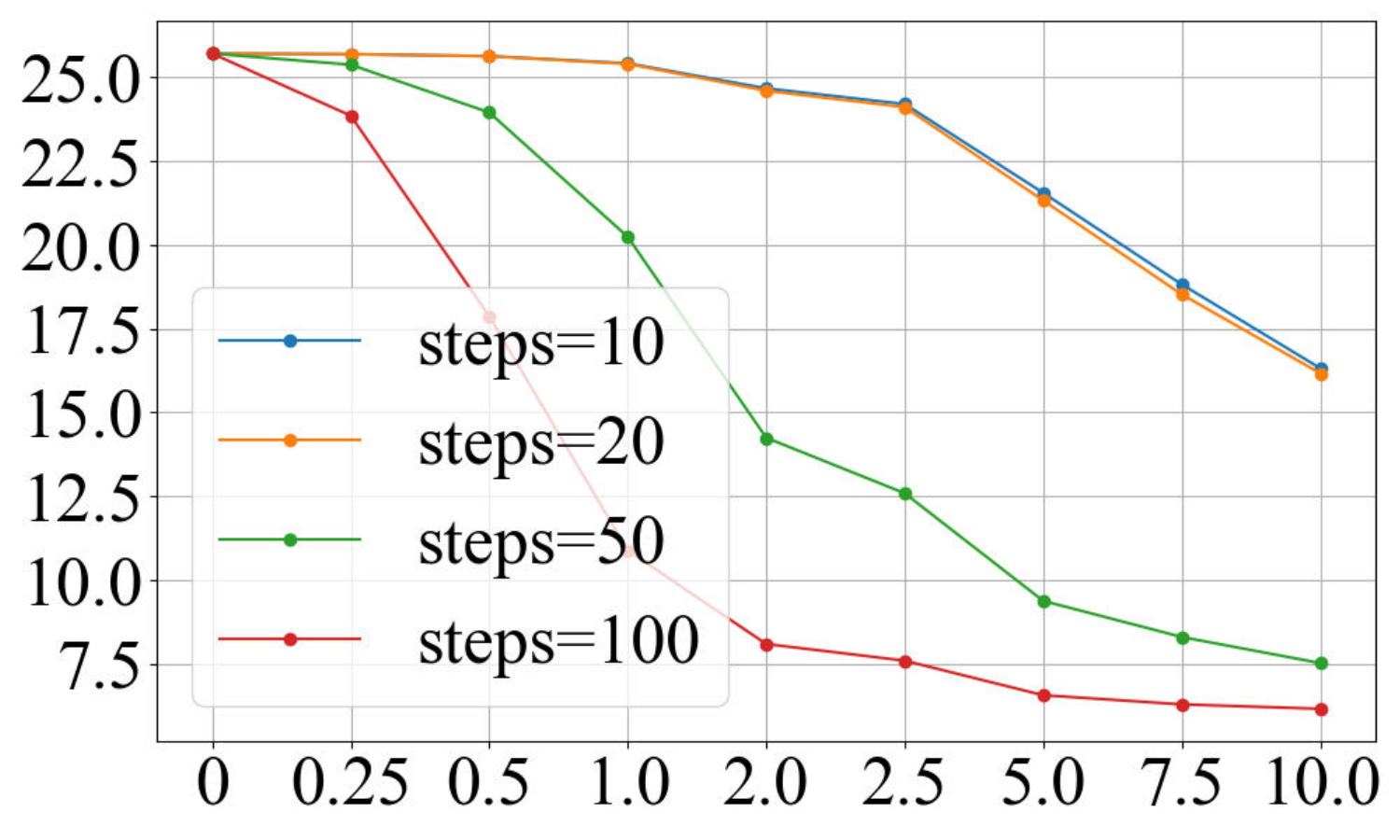}
        \caption{PSNR - $\sigma$}
    \end{subfigure}
    \begin{subfigure}{0.49\linewidth}
        \includegraphics[width=1.0\linewidth]{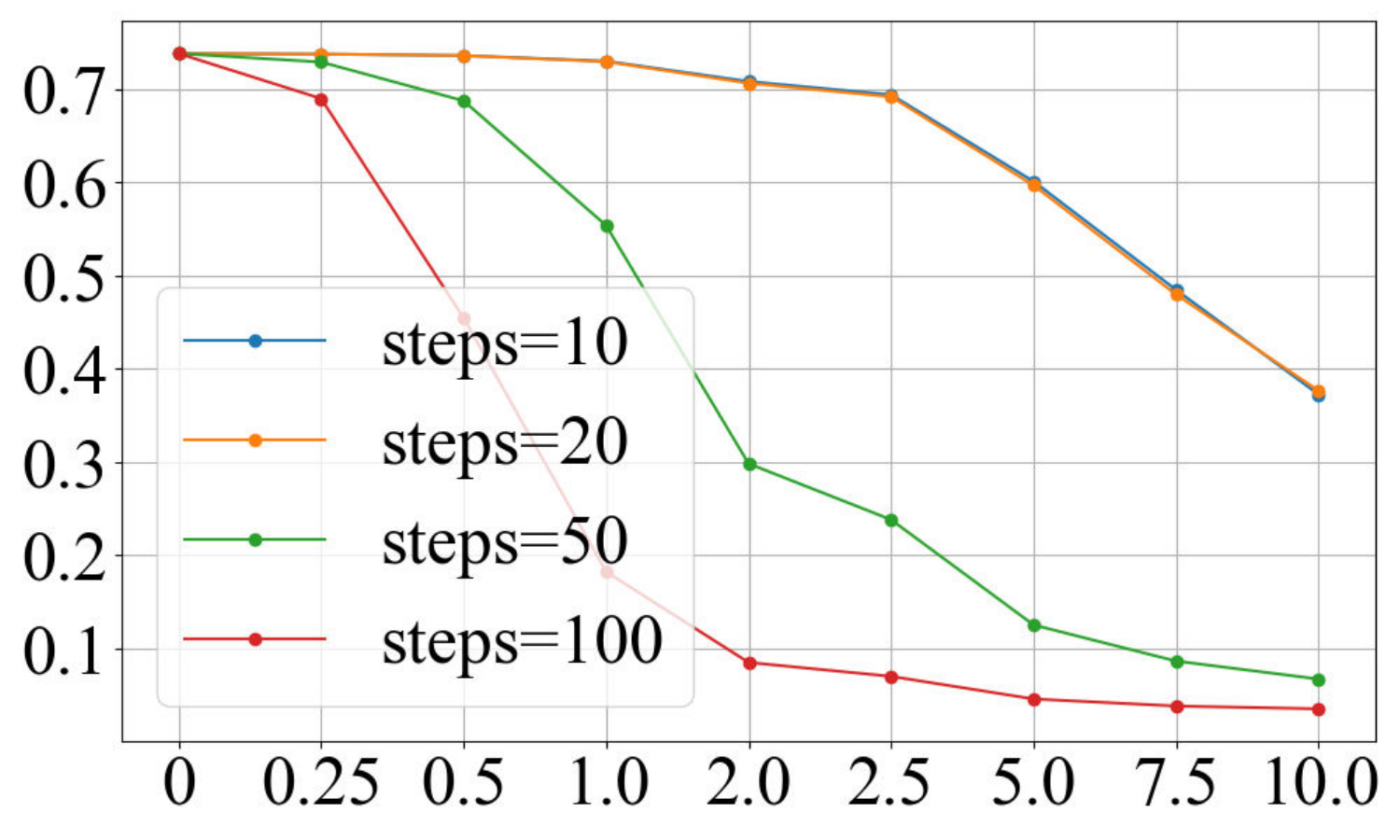}
        \caption{SSIM - $\sigma$}
    \end{subfigure}
    \begin{subfigure}{0.49\linewidth}
        \includegraphics[width=1.0\linewidth]{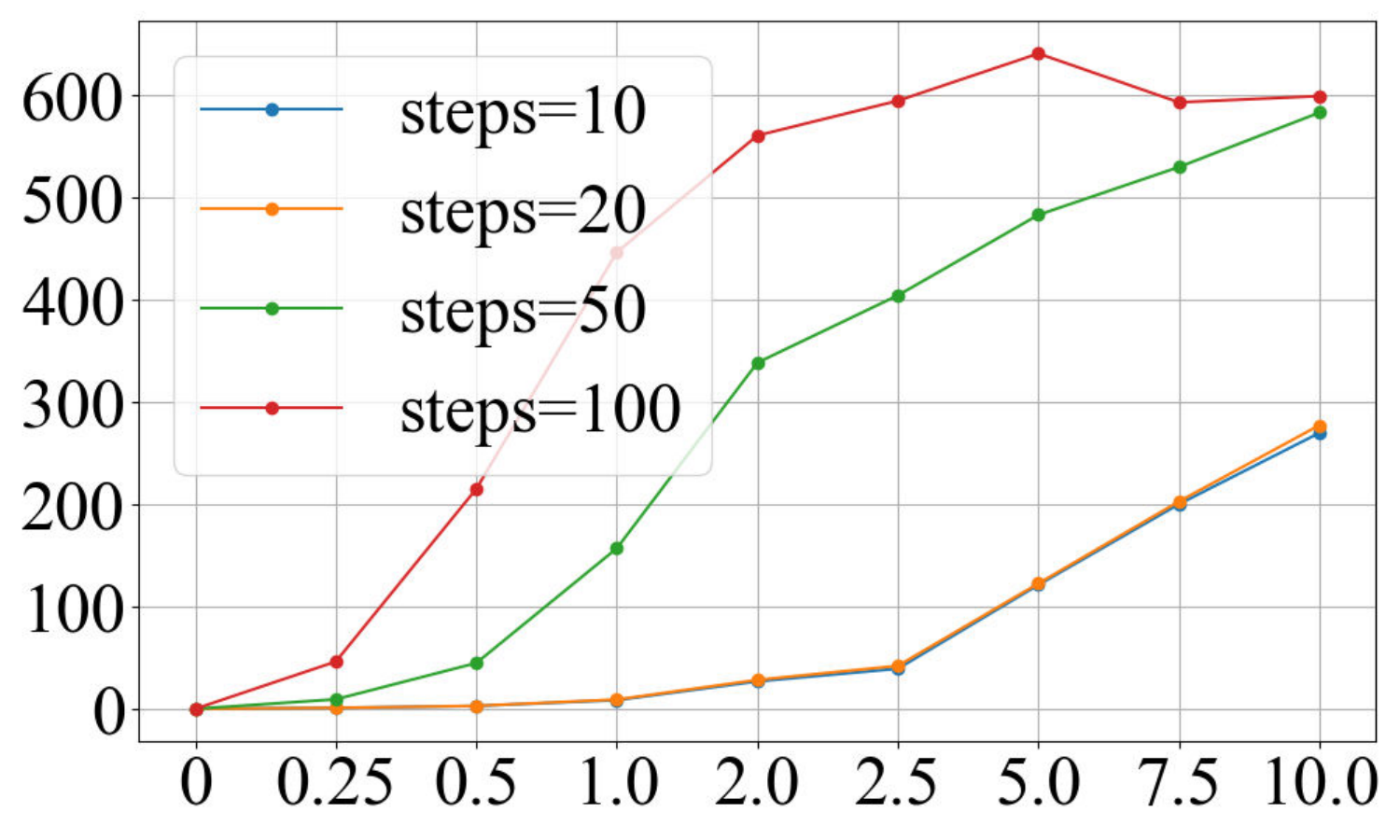}
        \caption{FID - $\sigma$}
    \end{subfigure}
    \begin{subfigure}{0.49\linewidth}
        \includegraphics[width=1.0\linewidth]{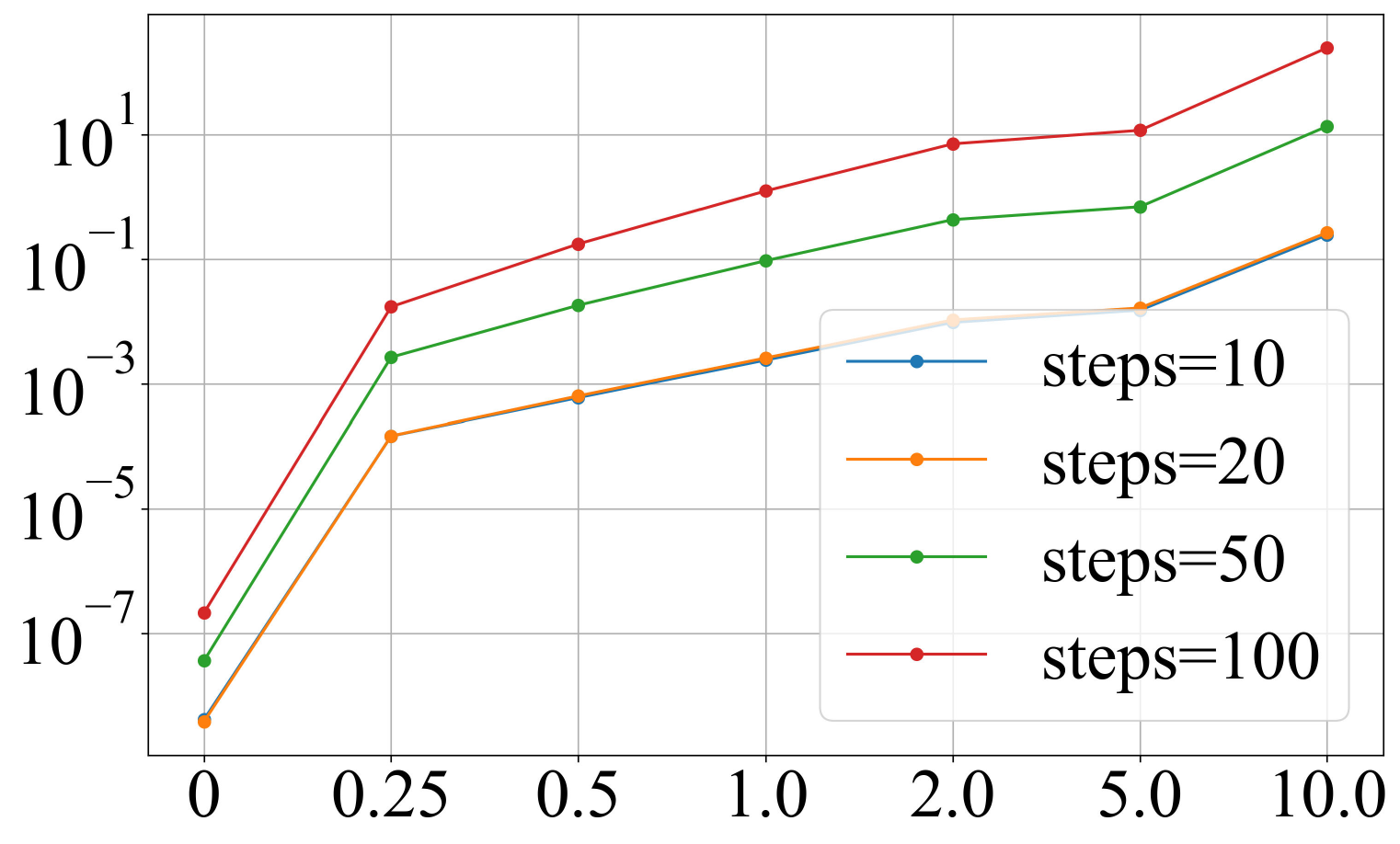}
        \caption{MSE - $\sigma$ ($\varepsilon$-correctness)}
    \end{subfigure}
    \caption{Variation of quality metrics with the noise variance $\sigma$}
    \label{fig:quality_curve}
\end{figure}
\begin{figure}
    \begin{subfigure}{0.325\linewidth}
        \includegraphics[width=1.0\linewidth]{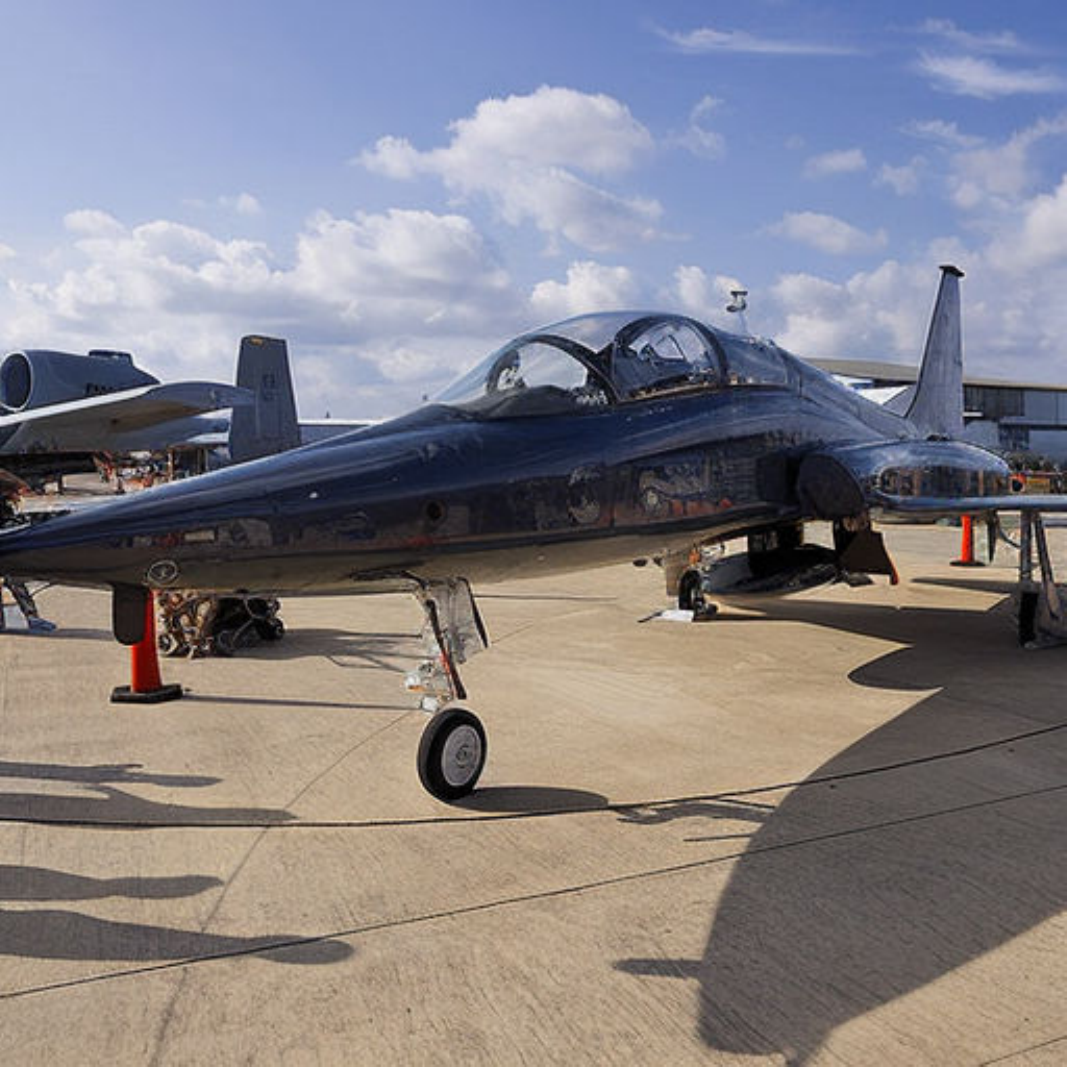}
        \caption{$\sigma=0$}
    \end{subfigure}
        \begin{subfigure}{0.325\linewidth}
        \includegraphics[width=1.0\linewidth]{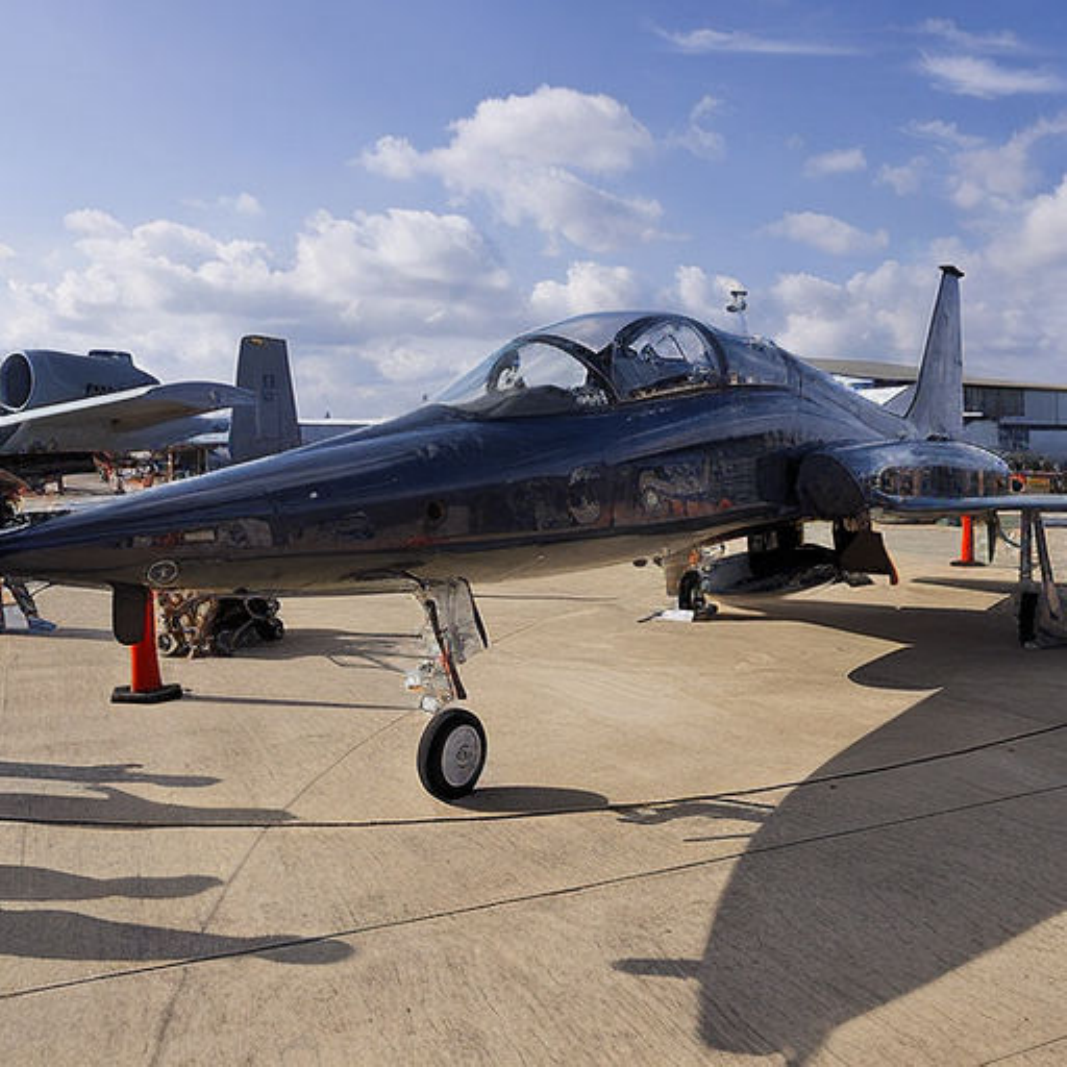}
        \caption{$\sigma=0.5$}
    \end{subfigure}
        \begin{subfigure}{0.325\linewidth}
        \includegraphics[width=1.0\linewidth]{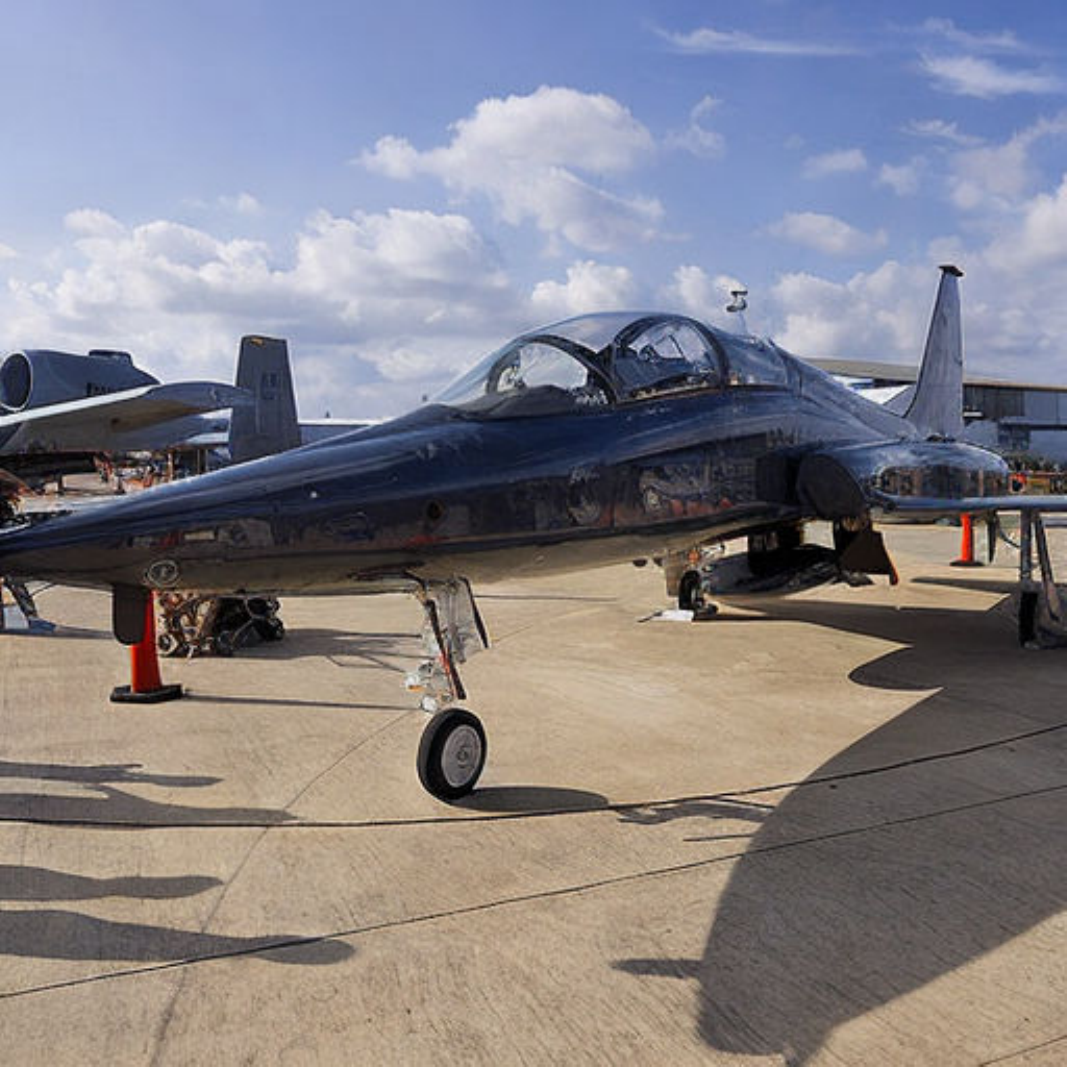}
        \caption{$\sigma=1$}
    \end{subfigure}
        \begin{subfigure}{0.325\linewidth}
        \includegraphics[width=1.0\linewidth]{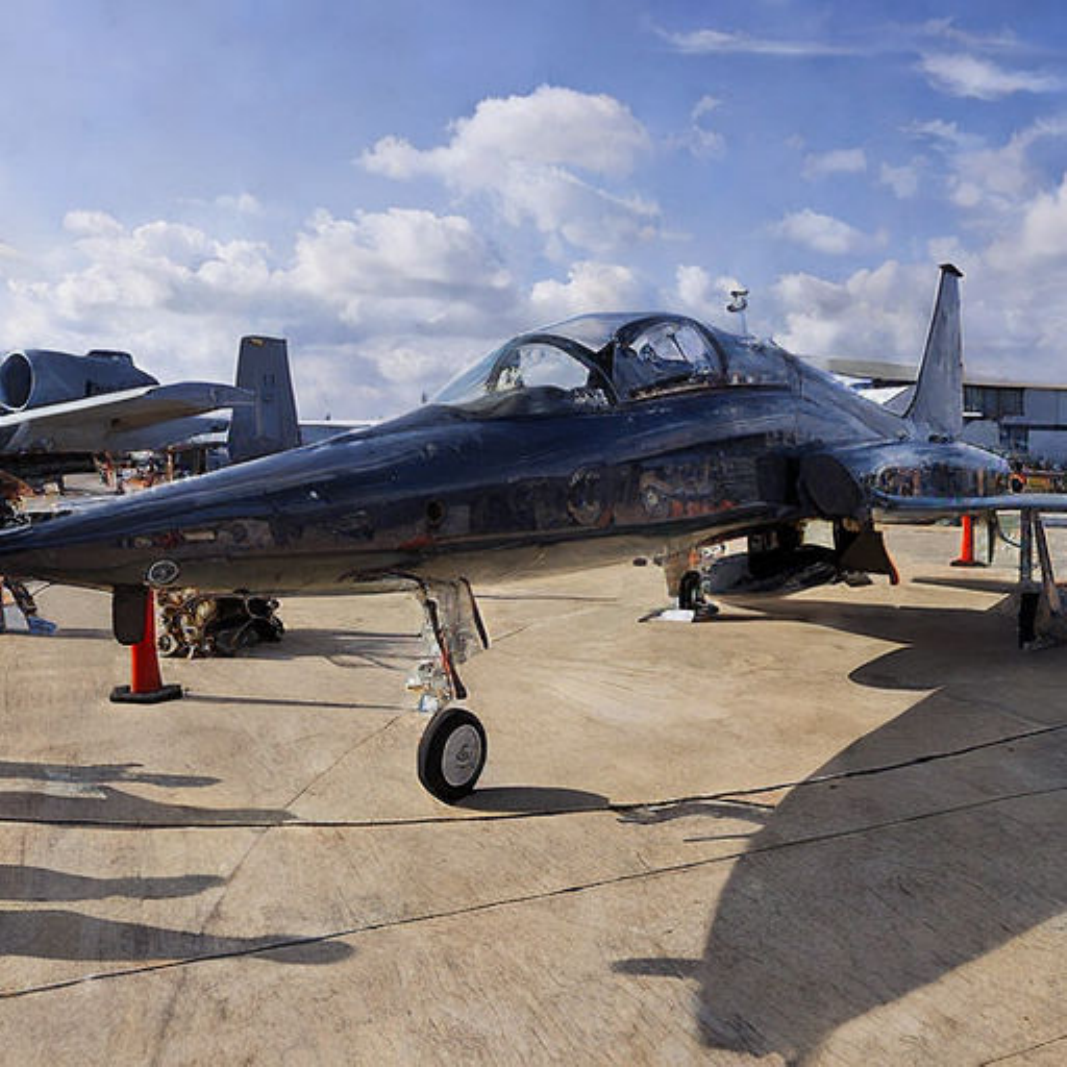}
        \caption{$\sigma=2$}
    \end{subfigure}
        \begin{subfigure}{0.325\linewidth}
        \includegraphics[width=1.0\linewidth]{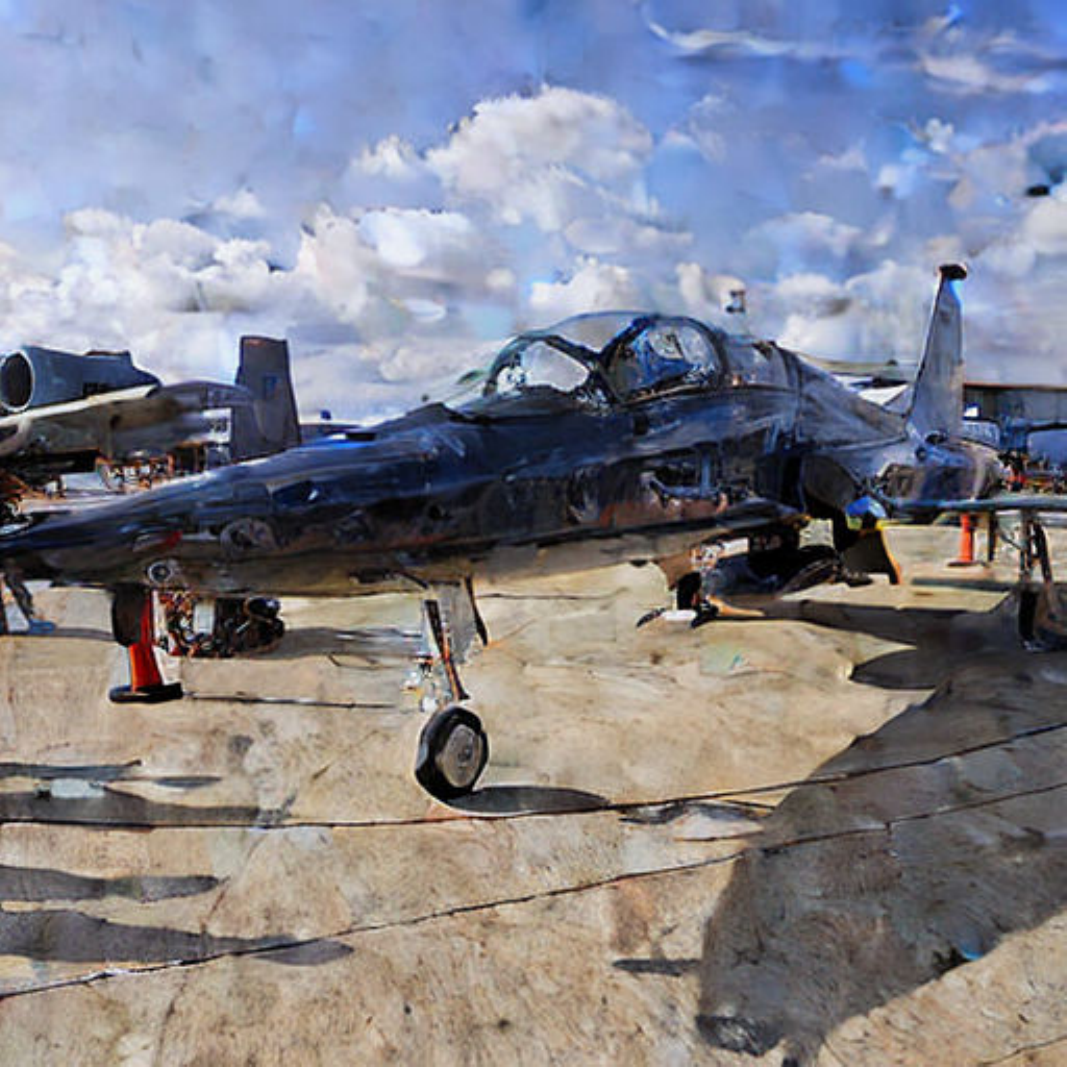}
        \caption{$\sigma=5$}
    \end{subfigure}
    \begin{subfigure}{0.325\linewidth}
        \includegraphics[width=1.0\linewidth]{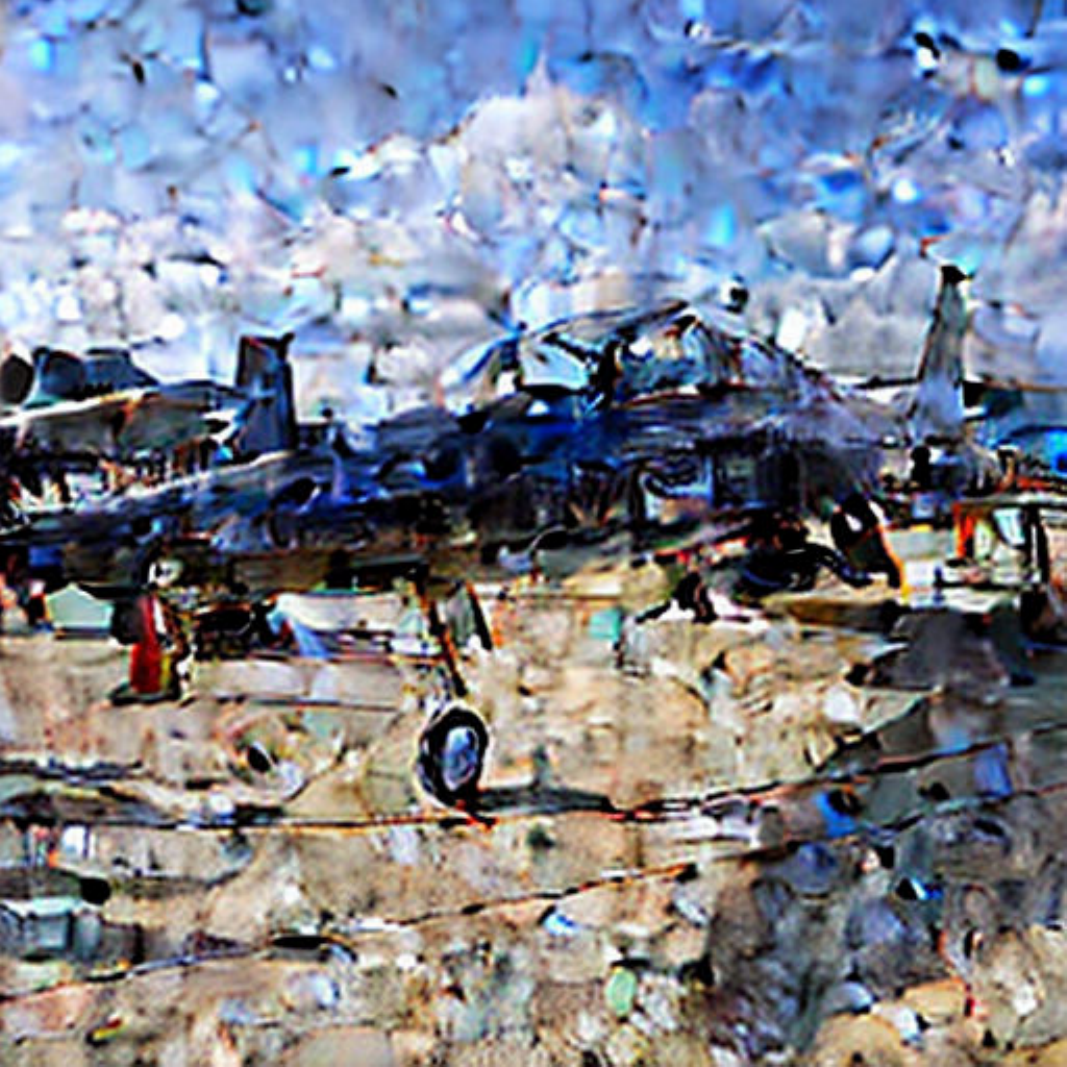}
        \caption{$\sigma=10$}
    \end{subfigure}
    \caption{Visual effects of quality degradation with the increase of $\sigma$. Here we hold sampling steps $T=20$.}
    \label{fig:quality_visual}
\end{figure}
This experiment evaluates the empirical $\varepsilon$-correctness of the framework by measuring the reconstruction quality under correct decryption across different sampling steps $T$ and noise intensities $\sigma^2$. The latent Mean Squared Error (MSE) is computed to estimate the practical reconstruction error bound $\varepsilon$. For experimental clarity, we set the key noise variance to a time-invariant constant $\sigma^2$ across all sampling steps. The baseline ($\sigma^2=0$) corresponds to standard O-BELM without protection. The quantitative trends across parameters are shown in Fig.~\ref{fig:quality_curve}. Example reconstructions under different noise levels are provided in Fig.~\ref{fig:quality_visual}.

Experimental results show that as the number of sampling steps and the noise intensity increase, the quality of the reconstructed images decreases, while the MSE between the latent variables increases, that is, the $\varepsilon$-correctness error coefficient $\varepsilon$ rises. This is caused by floating-point errors in computation: high-intensity noise loses low-order information during floating-point operations, and Sec.~\ref{sec_proof} indicates that multi-step inversion amplifies this low-order error to a level that noticeably affects reconstruction quality. We also note that even with $\sigma = 0$, perfectly reversible inversion is hardly attainable. Besides the floating-point errors mentioned above, the reconstruction error introduced by the VAE is another significant contributing factor.

\subsection{Cross-Model Robustness}
\label{cross_model_robust}
To assess the robustness against potential model updates or fine-tuning as a common scenario in practice, we evaluate its performance when encryption and decryption are performed with different diffusion model variants. We symmetrically test all pairwise combinations among four Stable Diffusion checkpoints (v1.4, v1.5, v2.0, and v2.1). For each pair (enc-model, dec-model), we generate ciphertexts and attempt reconstruction across the same ranges of step counts $T$ and noise intensities $\sigma^2$ as in Sec .~\ref {Qu_ConRe}. 

\begin{figure}[!ht]
    \begin{subfigure}{0.49\linewidth}
        \includegraphics[width=1.0\linewidth]{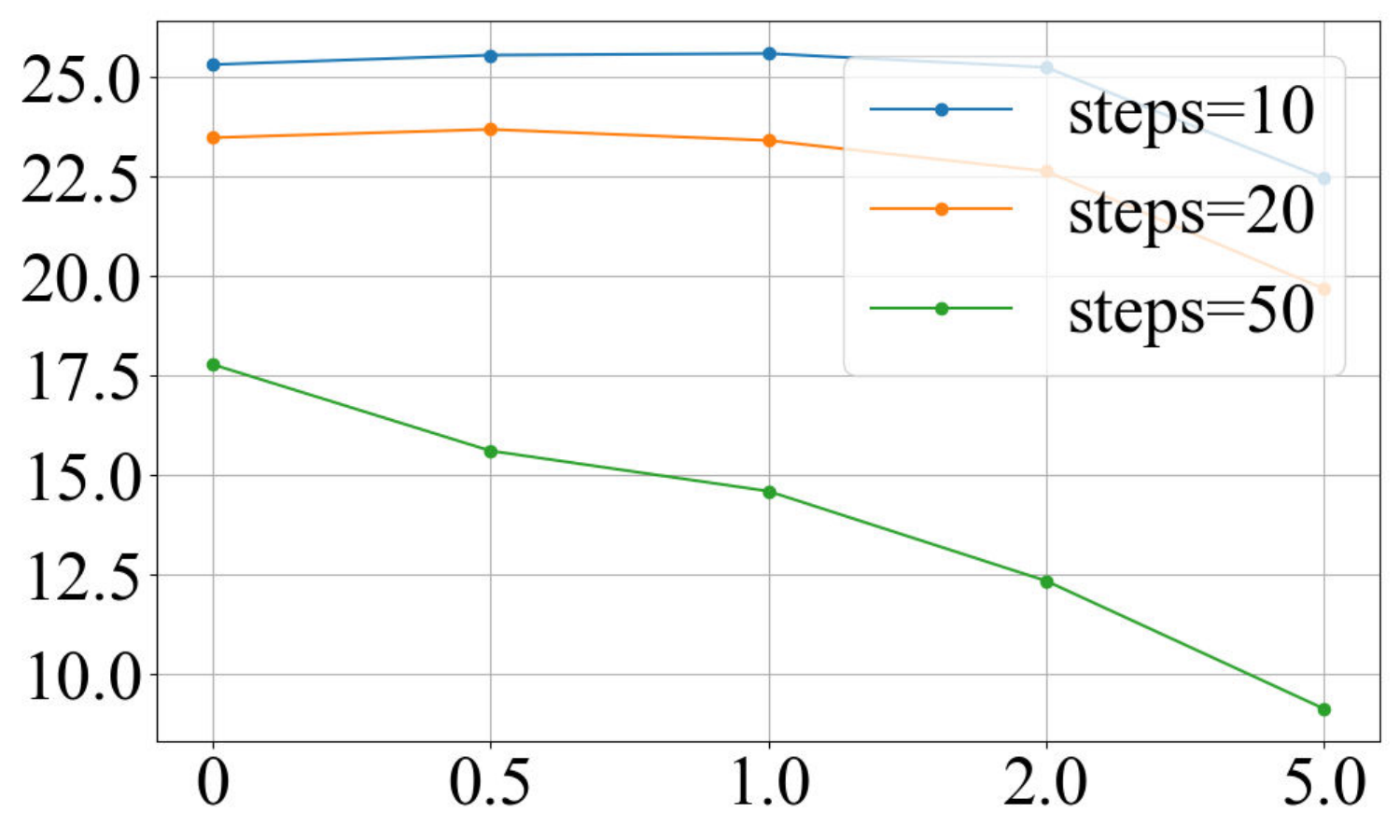}
    \end{subfigure}
    \begin{subfigure}{0.49\linewidth}
        \includegraphics[width=1.0\linewidth]{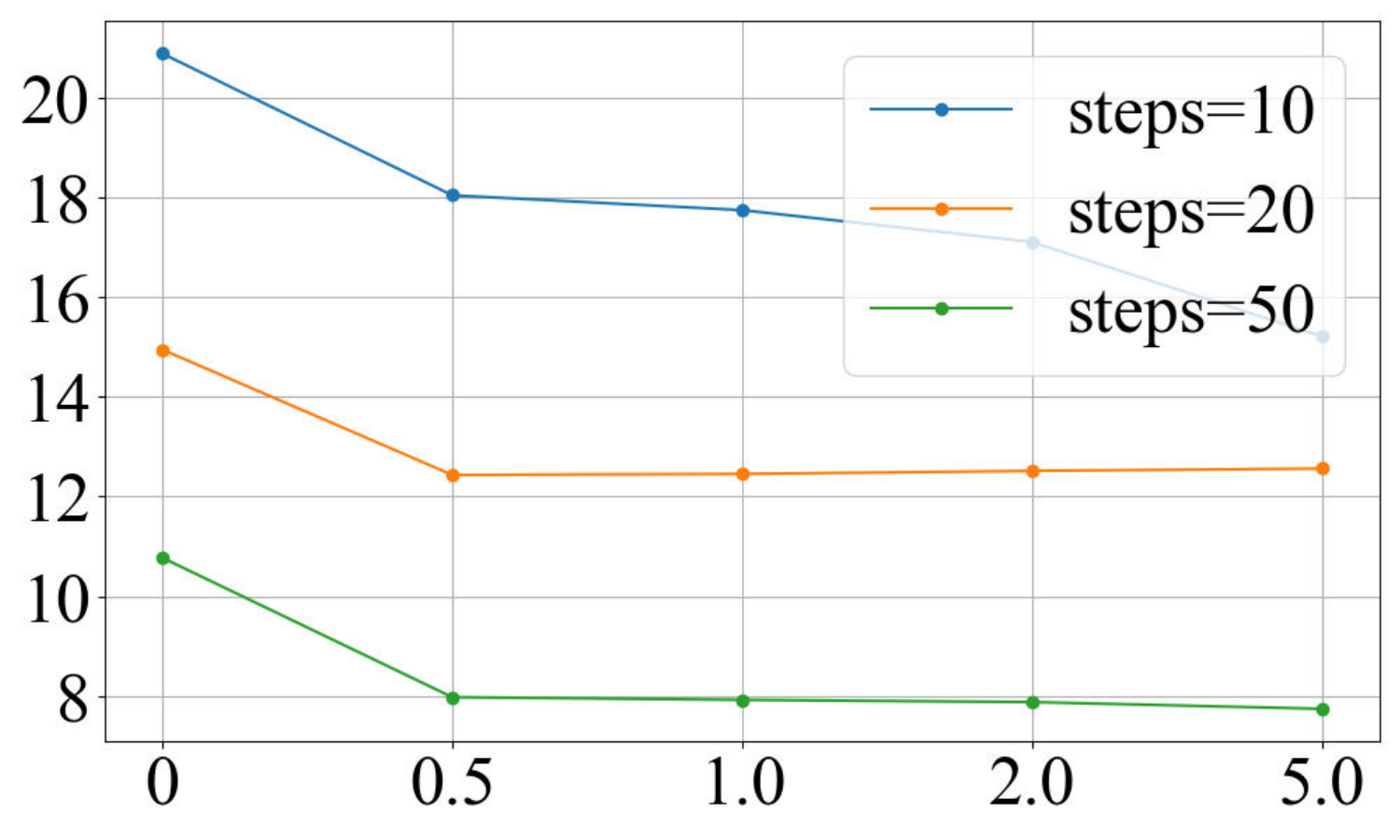}
    \end{subfigure}
    \begin{subfigure}{0.49\linewidth}
        \includegraphics[width=1.0\linewidth]{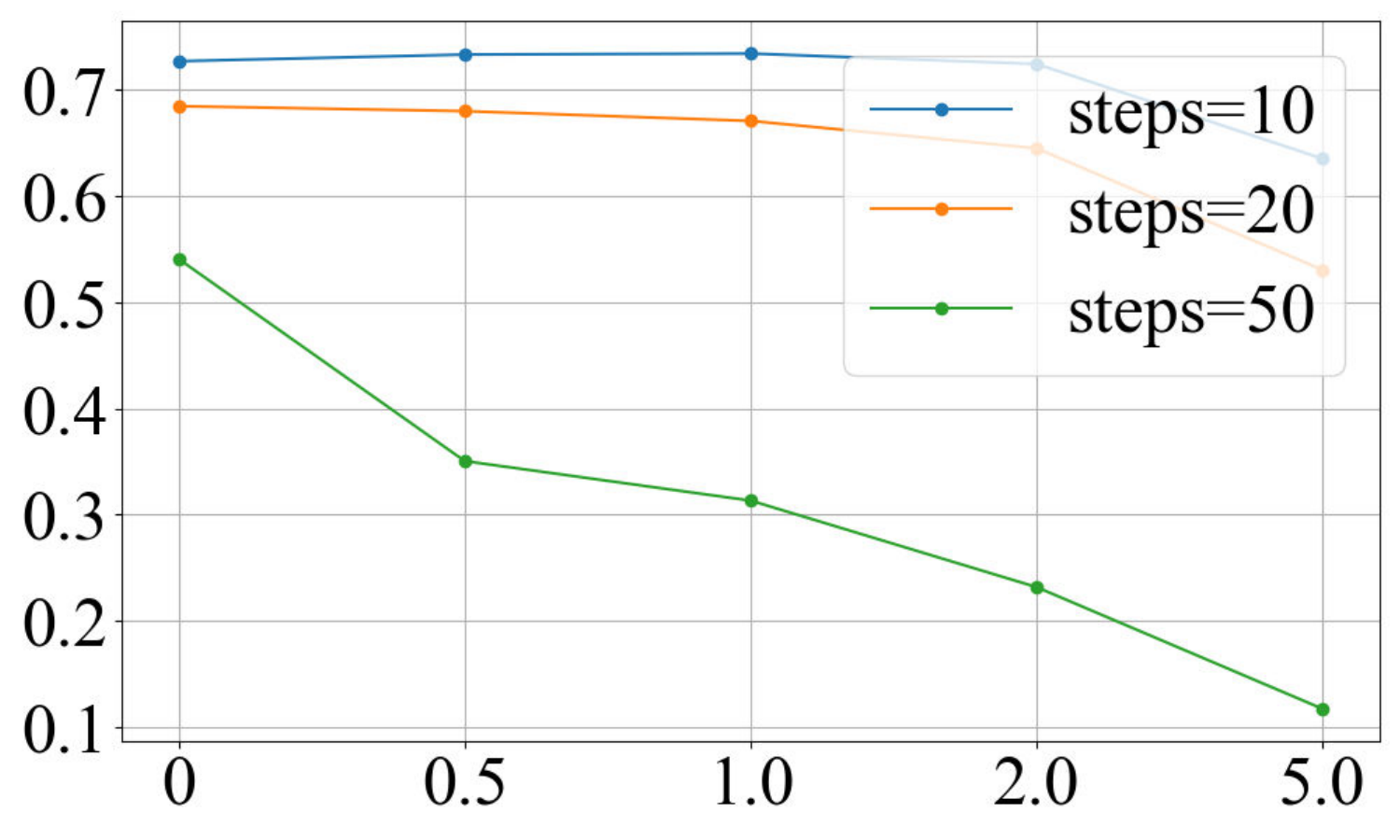}
    \end{subfigure}
        \begin{subfigure}{0.49\linewidth}
        \includegraphics[width=1.0\linewidth]{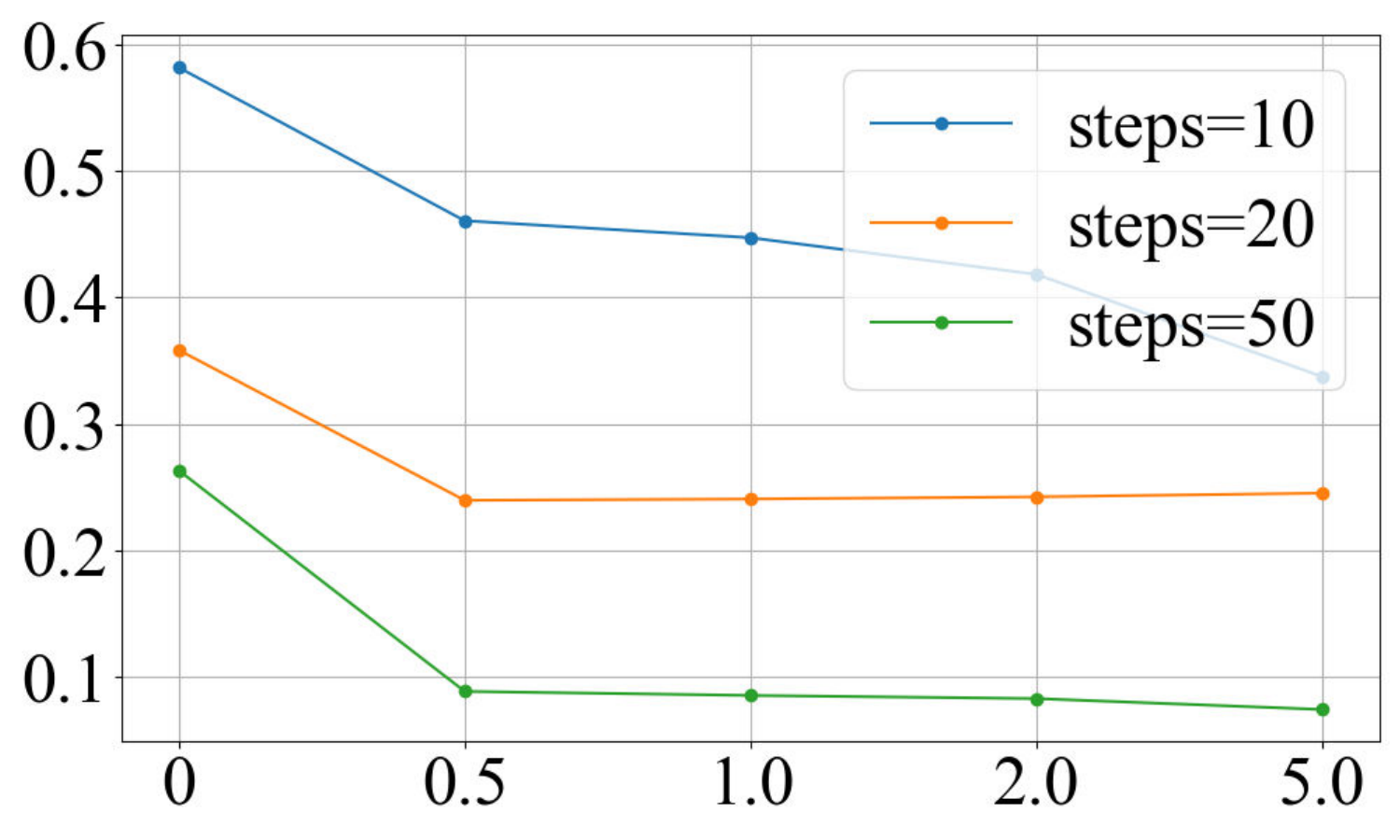}
    \end{subfigure}
        \begin{subfigure}{0.49\linewidth}
        \includegraphics[width=1.0\linewidth]{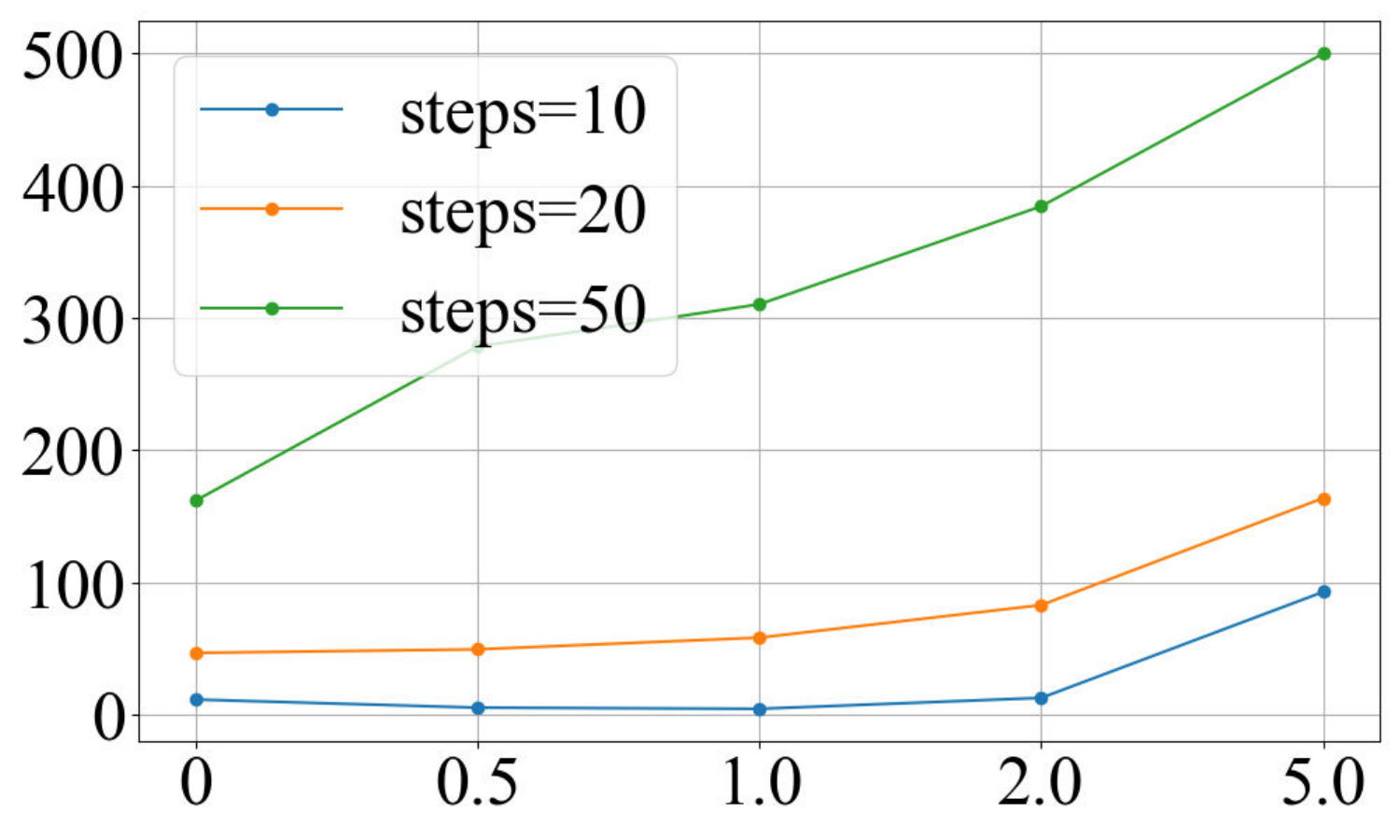}
        \caption{Near-version models}
    \end{subfigure}
        \begin{subfigure}{0.49\linewidth}
        \includegraphics[width=1.0\linewidth]{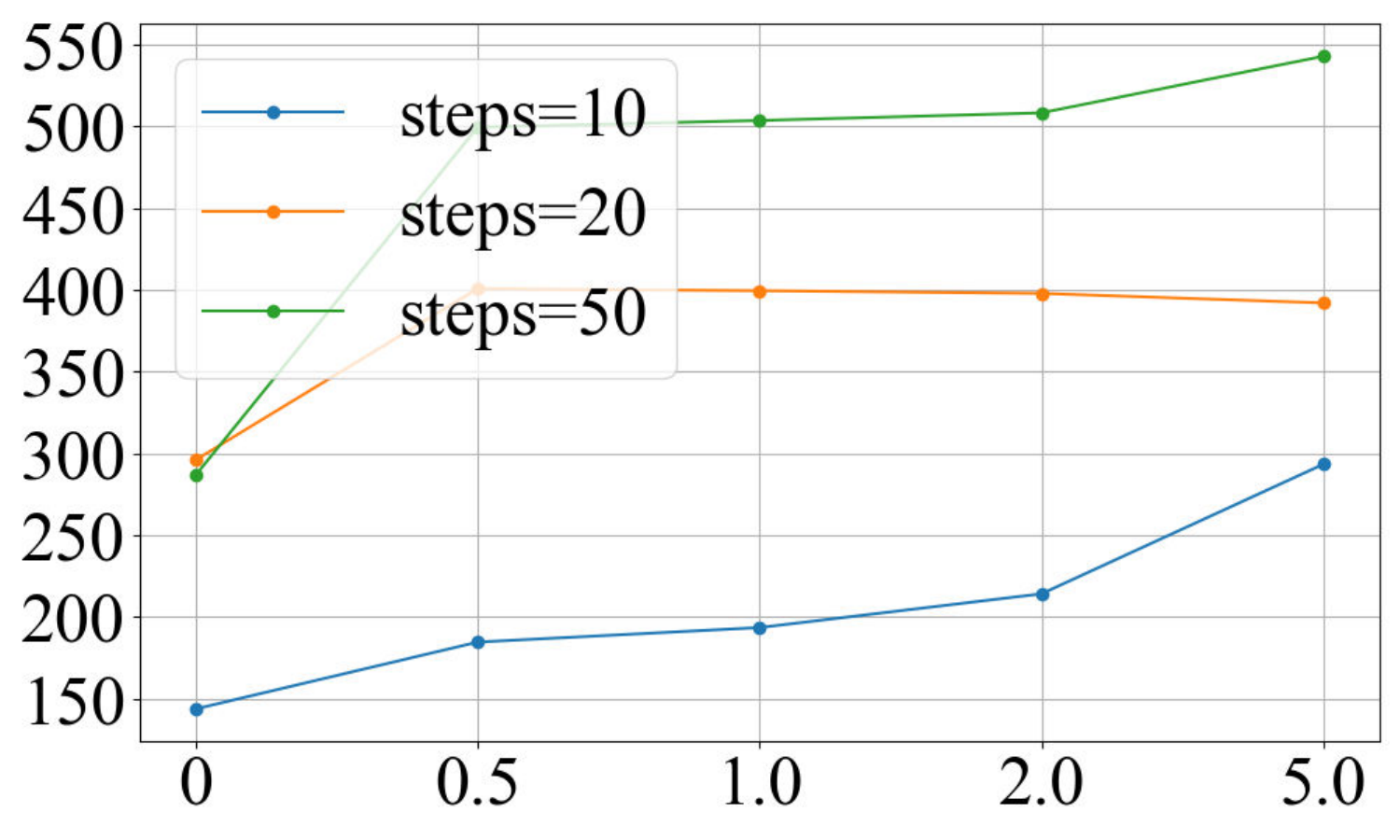}
        \caption{Distant-version models}
    \end{subfigure}
    \caption{Evaluation of cross‑model robustness. As in the Sec.~\ref{Qu_ConRe}, we adopt three metrics: PSNR (top row), SSIM (middle row), and FID (bottom row).}
    \label{fig:cross_model_eva}
\end{figure}
\begin{figure}
    \centering
    \includegraphics[width=0.65\linewidth]{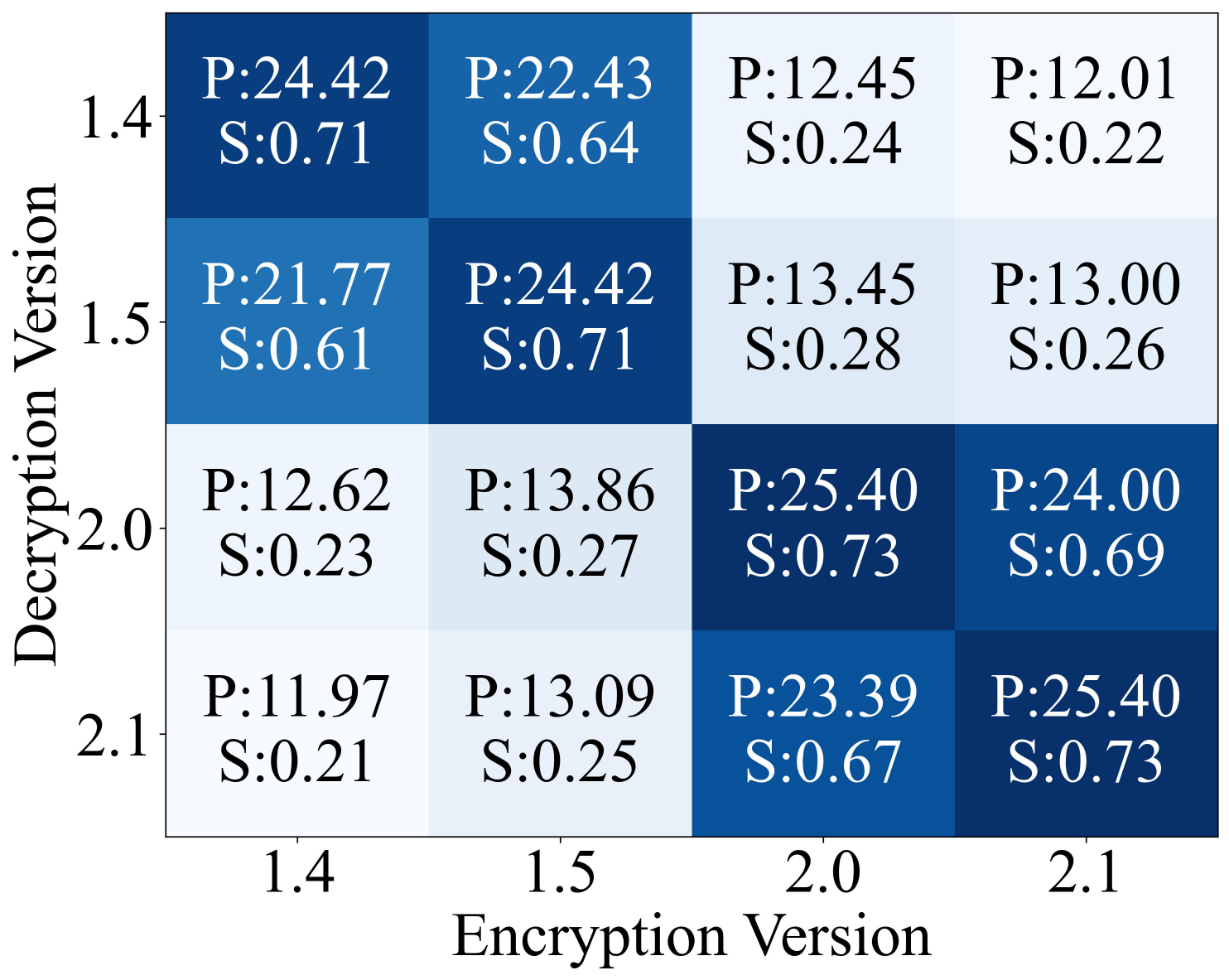}
    \caption{The heatmap of cross-model robustness, where P indicates PSNR and S indicates SSIM.}
    \label{fig:cross_model_heatmap}
\end{figure}
\begin{table}
    \centering
    \begin{tabular}{lcccc}
        \toprule
       Method & Condition & $\text{PSNR/dB}$& $\text{SSIM}$& $\text{FID}$ \\
        \midrule
        \multirow{3}{*}{\textbf{PCDiff}} & corr. key & 30.13 & 0.91 & 8.54\\
        & wr. key  & 25.35 & 0.85 & 34.81\\
        & w/o key & 21.49 & 0.82 & 39.37\\
        \multirow{3}{*}{\textbf{RNDMask}} & corr. key & 25.51 & 0.76 & 27.39\\
        & wr. key & 13.37 & 0.35 & 299.68\\
        & w/o key & 7.51 & 0.08 & 423.61\\
        \multirow{3}{*}{\textbf{Our Scheme}} & corr. key & 25.41 & 0.73 & 9.05\\
        & wr. key & \textbf{5.84} & \textbf{0.03} & \textbf{576.90}\\
        & w/o key & \textbf{5.84} & \textbf{0.03} & \textbf{640.98}\\
        \bottomrule
    \end{tabular}
    \caption{Conditional Access Control with different existing methods, where for wrong-key (wr. key) and no-key (w/o key) settings the desired trends are opposite to those for the correct key (corr. key): \textbf{lower} PSNR/SSIM and \textbf{higher} FID indicate stronger protection.}
    \label{tab:acc_control}
\end{table}
The results are aggregated into Fig.~\ref{fig:cross_model_eva}, and visualized as a heatmap (Fig.~\ref{fig:cross_model_heatmap}) for a representative $T=20$ and $\sigma=1$.  We observe that when using closely related model versions (e.g., encrypting with Stable Diffusion v2.0 and decrypting with v2.1), the robustness performance is better compared to using versions with larger differences.

As a controlled comparison, we also test reconstruction from ciphertexts perturbed by a small amount of additive noise but using the same model for both encryption and decryption, resulting in $\text{PSNR}\approx6.41 \text{dB}$ and $\text{FID}\approx606.51$. The result is irrelevant to the model version. Thus, the impact of cross-model robustness is empirically isolated from general ciphertext sensitivity.  

Finally, we provide a theoretical analysis of cross‑model robustness based on the error‑propagation dynamics in the decryption/sampling path. Details of this analysis are discussed in Sec.~\ref{cross_model_analysis}.

\subsection{Security Evaluation}
\label{sec_eva}
\begin{table}
    \centering
    \begin{tabular}{cccc}
        \toprule
        Steps & Avg. Cost (s) & $\mathsf{Adv}^{\mathrm{IND-CPA}}_{\mathcal{A}}$ & Attack time (h)\\
        \midrule
        10 & 1.01 & $3.44 \times10^{-11}$ & $4.08 \times 10^6$\\
        20 & 1.85 & $6.98 \times 10^{-11}$ & $3.68 \times 10^6$\\
        50 & 4.35 & $1.95 \times 10^{-10}$ & $3.10 \times 10^6$\\
        100 & 8.58 & $3.91\times 10^{-10}$ & $3.04 \times 10^6$\\
        200 & 16.92 & $6.17\times10^{-10}$ & $3.80 \times 10^6$\\
        \bottomrule
    \end{tabular}
    \caption{Encryption-Decryption Time}
    \label{enc_dec_time}
\end{table}
This section provides a comprehensive empirical assessment of the security properties, designed to complement the theoretical analysis. The evaluation consists of three targeted sub-experiments:

\textbf{Perceptual Security:} We compare reconstruction quality under three conditions: with the correct key (corr. key), with a wrong key (wr. key, under single bit flip), and without any key (w/o key). Results are benchmarked against representative existing methods for image encryption or access control in diffusion models(e.g., PCDiff~\cite{gai2025pcdiff}, RNDMask~\cite{Takana2025access}).

As verified in Tab.~\ref{tab:acc_control}, due to the performance limitations of the VAE and the presence of floating-point numerical errors, the quality of correct reconstruction is constrained. However, under both wrong key and no key conditions, our method prevents unauthorized reconstruction by an adversary more effectively compared to the other two approaches.

\textbf{Validation of Theoretical Bounds:} To empirically support Asm.~\ref{asm:linear_growth}, we measure the norm of the error propagation operator $\|\mathbf{\Phi}(T,t)\|$ and the square of the eigenvalue of the ciphertext variance matrix $\lambda^2(\mathbf{\Sigma}_T)$ across different noise schedules and maximum step counts $N$, under a controlled setting where the model's noise prediction $c_i\mathbf{E_i}$ is disabled.
\begin{figure}[!t]
    \begin{subfigure}{0.49\linewidth}
        \includegraphics[width=1.0\linewidth]{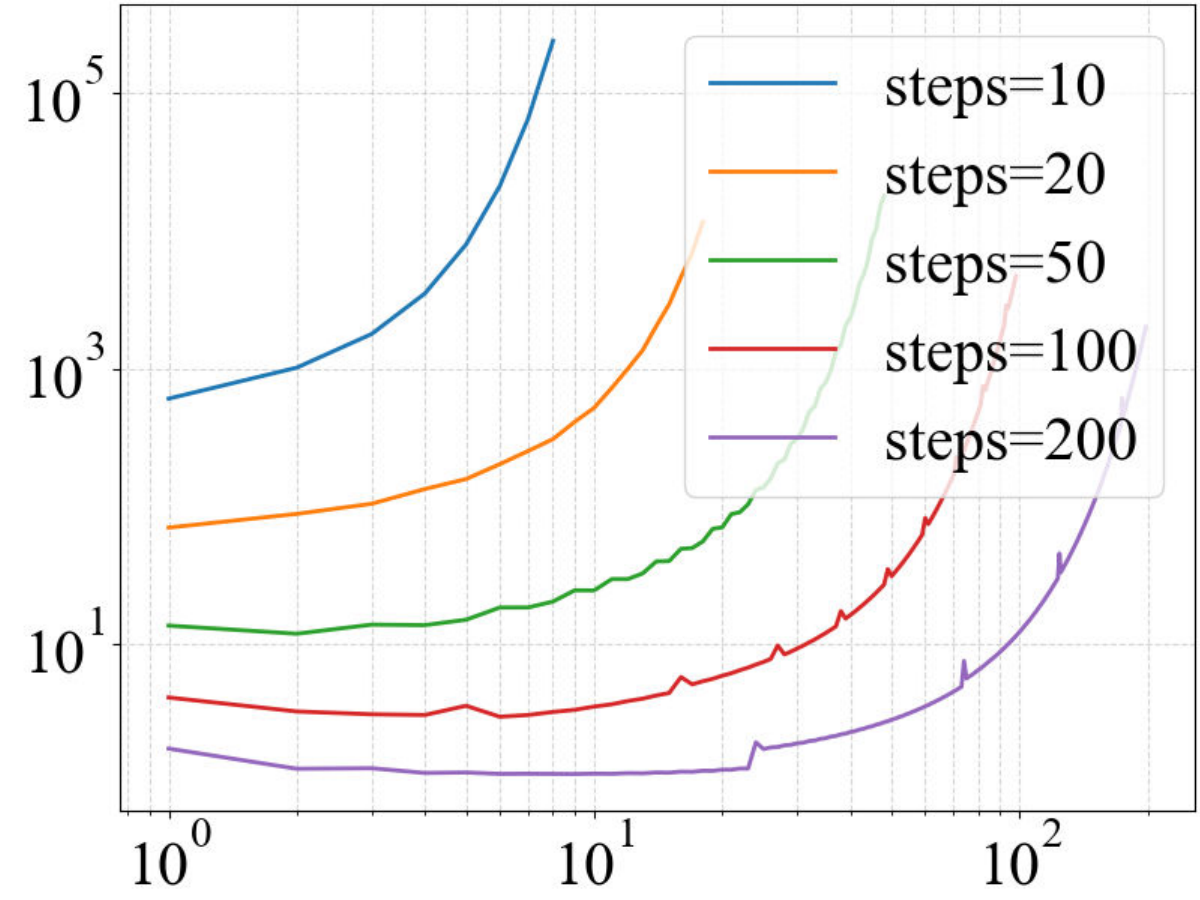}
    \end{subfigure}
    \begin{subfigure}{0.49\linewidth}
        \includegraphics[width=1.0\linewidth]{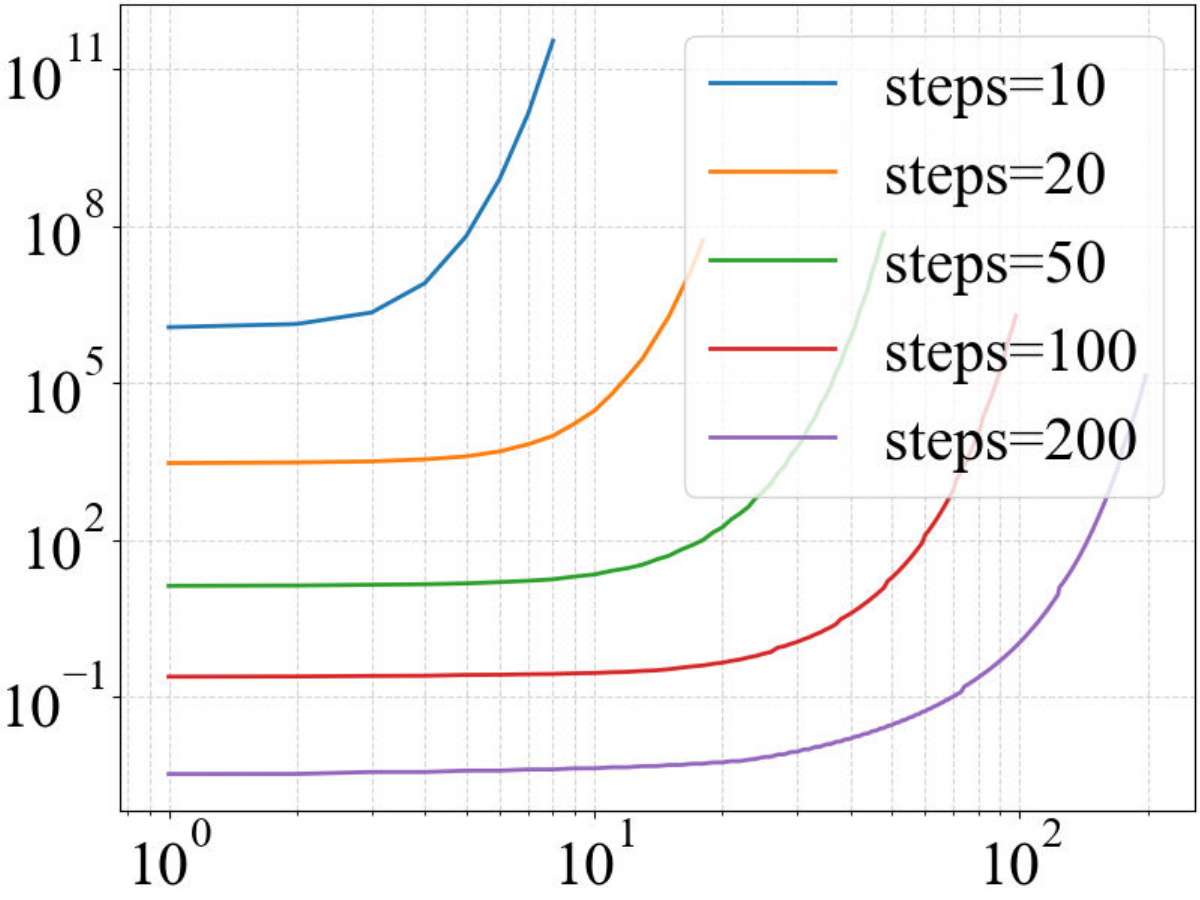}
    \end{subfigure}
    \begin{subfigure}{0.49\linewidth}
        \includegraphics[width=1.0\linewidth]{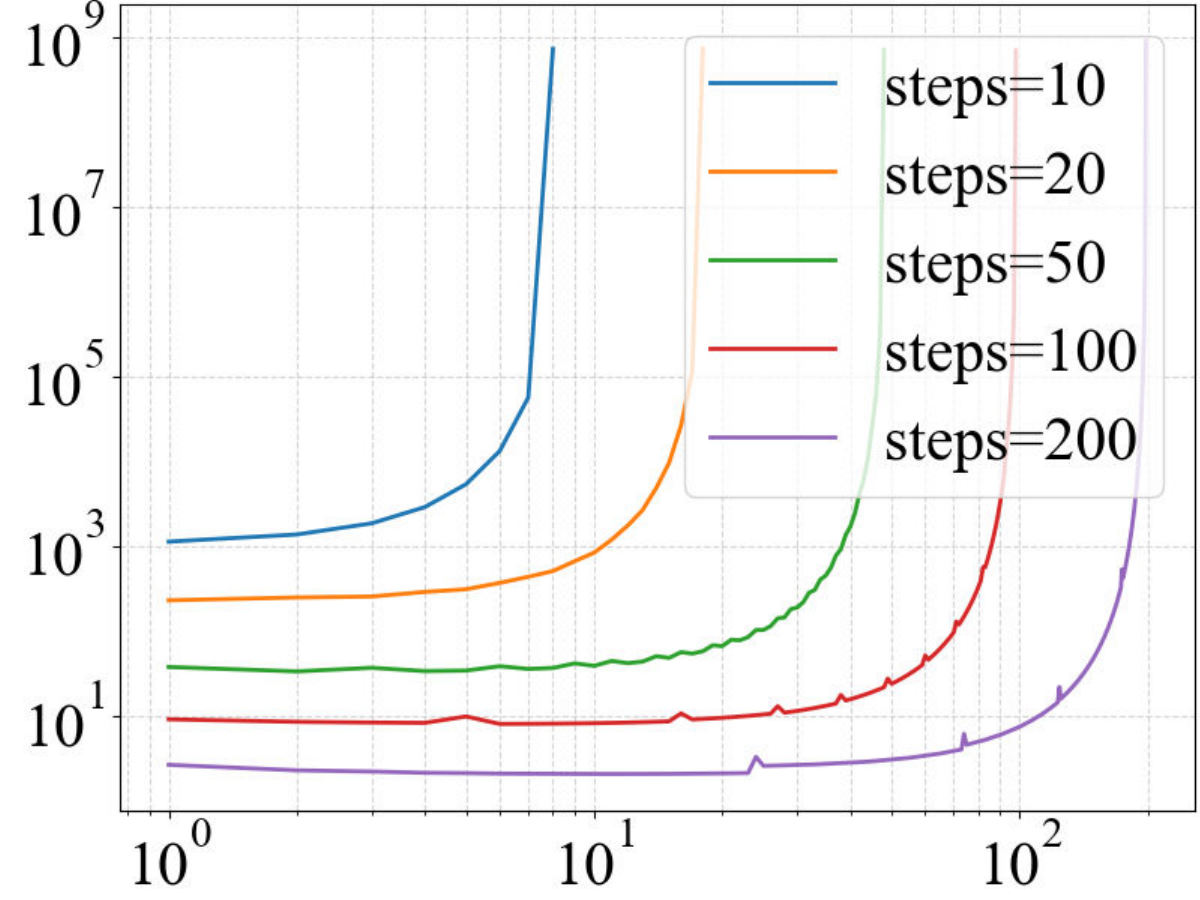}
    \end{subfigure}
    \begin{subfigure}{0.49\linewidth}
        \includegraphics[width=1.0\linewidth]{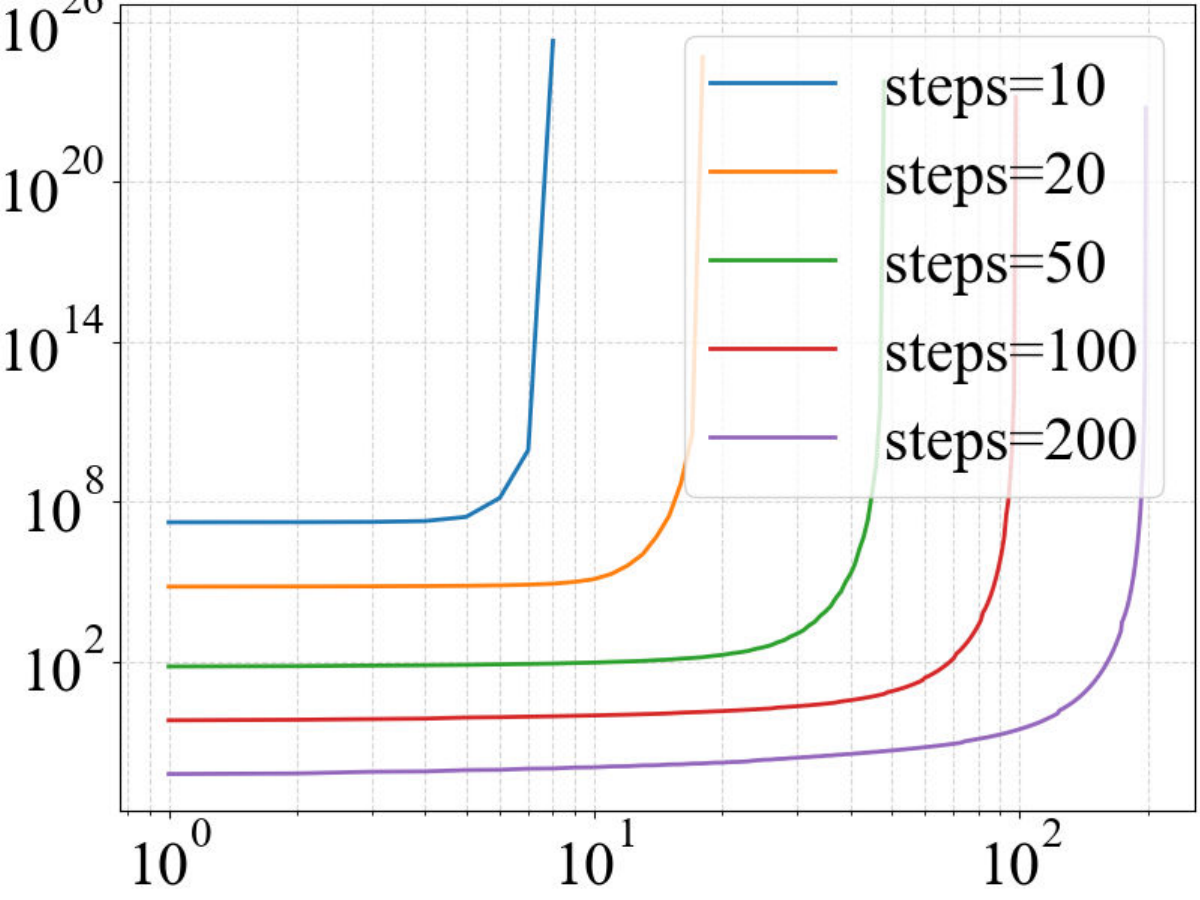}
    \end{subfigure}
    \begin{subfigure}{0.49\linewidth}
        \includegraphics[width=1.0\linewidth]{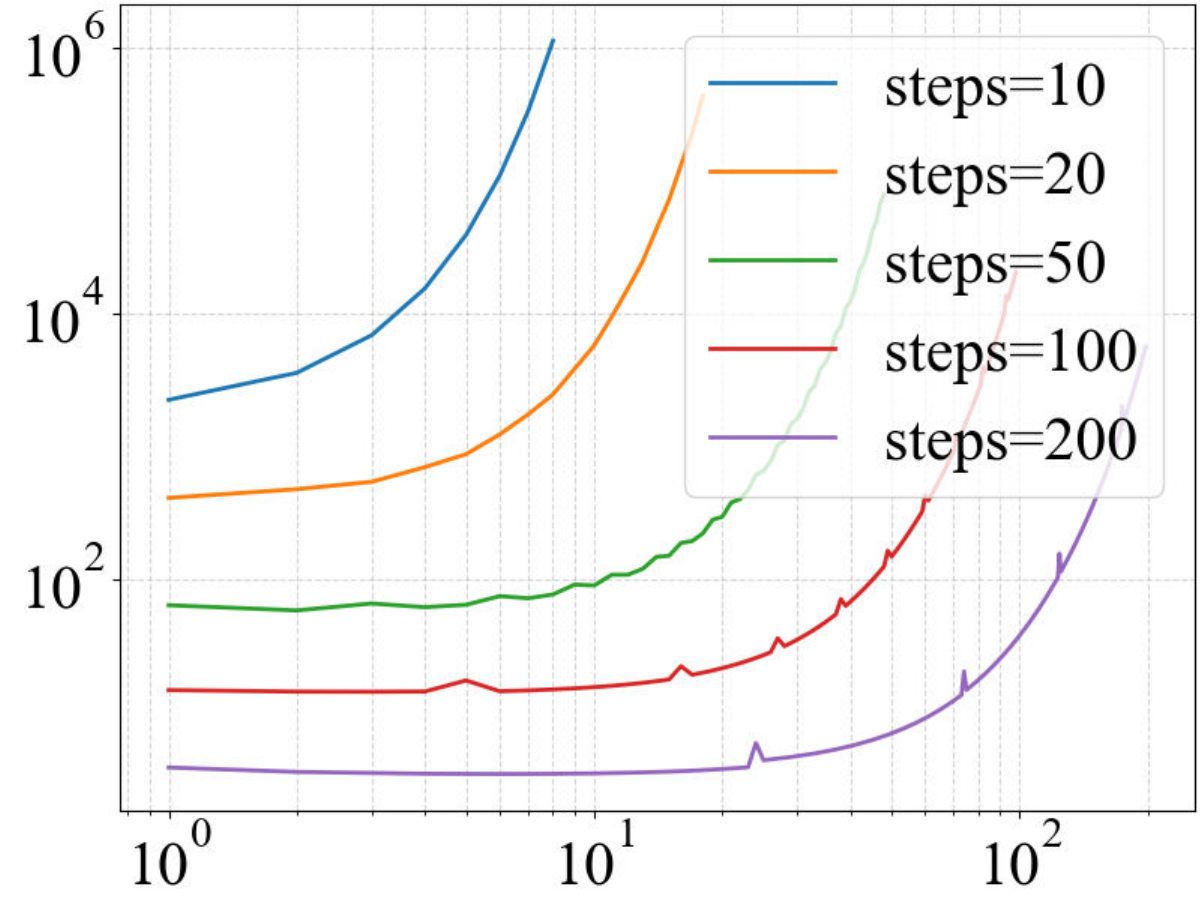}
        \caption{$\|\mathbf{\Phi}(T,t)\|$}
    \end{subfigure}
    \begin{subfigure}{0.49\linewidth}
        \includegraphics[width=1.0\linewidth]{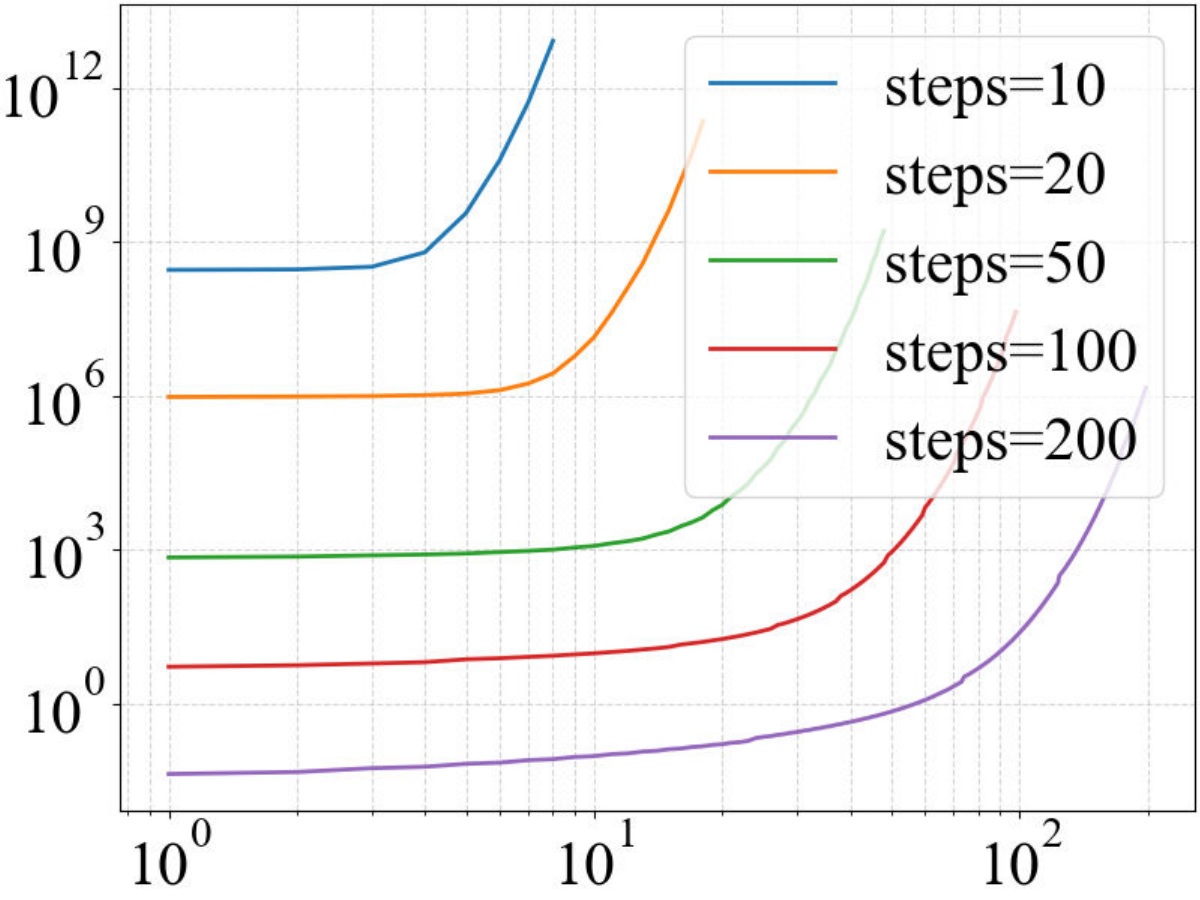}
        \caption{$\lambda^2(\mathbf{\Sigma}_T)$}
    \end{subfigure}
    \caption{Parameter trends across sampling steps, where each row corresponds to a different scheduling method: squared linear (first row), squared cos (second row), linear (third row).}
    \label{fig:phi_sigma}
\end{figure}

As shown in Fig.~\ref{fig:phi_sigma}, the two parameters exhibit a highly consistent growth pattern: they remain relatively stable in the early sampling stages and then undergo a transition to exponential growth later. Based on these observations, we conclude that Asm.~\ref{asm:linear_growth} is well-founded.

If Assumption~\ref{asm:linear_growth} holds, the ciphertext distributions for two different plaintexts are approximately Gaussian with the same covariance matrix \(\boldsymbol{\Sigma}_T\) but different means \(\mathbf{x}_T^{(0,0)}\) and \(\mathbf{x}_T^{(0,1)}\). In this case, the squared Mahalanobis distance between the two distributions is given in Eq.~\ref{maha_dist}. Consequently, as proved in Appendix~\ref{prof_1}, $d_{\text{CT}}^2$ satisfies the following upper bound:
\[
d_{\text{CT}}^2 \leq \frac{36\dim(\mathbf{x}_0)}{\lambda_{\min}},
\]
where $\lambda_{\min}$ denotes the smallest eigenvalue of $\boldsymbol{\Sigma}_T$.

Thus, by estimating $\lambda_{\min}(\boldsymbol{\Sigma}_T)$ from linearized error propagation, we can compute a concrete upper bound on the distinguishability of the two ciphertext distributions. Meanwhile, according to Eq.~\eqref{sigma_sat}, the variance of the injected noise $\sigma_t^2$ directly influences $\boldsymbol{\Sigma}_T$: increasing $\sigma_t$ enlarges $\lambda_{\min}(\boldsymbol{\Sigma}_T)$ quadratically, thereby reducing the adversary's advantage at a quadratic rate. This provides an additional tunable parameter for security.

\textbf{Computational Security \& Practical Implications:}
We report the wall-clock time for a full encryption-decryption cycle in Tab.~\ref{enc_dec_time}. Combining this with the theoretical upper bound for the adversary's IND-CPA advantage above, we estimate the computational effort required for a successful attack through the equation below:
\begin{equation}
    T = \frac{1}{\mathsf{Adv}^{\mathrm{IND-CPA}}_{\mathcal{A},\Pi}} \times \frac{t_{C}}{2}
\end{equation}
where $t_C$ is the duration of one encryption–decryption cycle. An adversary mounting an attack would need to spend at least $t_C/2$ per query to the encryption oracle.

Although experiments show that under our given computational constraints, an attack would require at least $3\times 10^6 \text{h}\approx \mathbf{342.47} \text{years}$, from a long-term perspective, the system presents a latent risk if an adversary's computational resources are sufficiently advanced. This allows us to derive practical guidelines for setting access frequency limits and key refresh policies to maintain security in deployed scenarios.

\begin{table}
    \centering
    \begin{tabular}{cccccc}
        \toprule
        \multirow{2}{*}{\textbf{Steps}} & \multicolumn{5}{c}{Models \& Methods}\\
        & Ours & O-BELM\cite{wang2024belm} & DDIM\cite{ho2020denoising} &\cite{chen2024privacy} &\cite{he2025private1}\\
        \midrule
        10 & \textbf{1.01} & 0.85 & 0.85 & 14.82 & 2030.15\\
        20 & \textbf{1.85} & 1.21 & 1.21 & 31.95 & 3220.91\\
        50 & \textbf{4.35} & 2.48 & 2.47 & 70.92 & 8352.14\\
        100 & \textbf{8.58} & 4.71 & 4.70 & 137.66 & $>$10000\\
        200 & \textbf{16.92} & 9.05 & 9.03 & 262.10 & $>$10000\\
        \bottomrule
    \end{tabular}
    \caption{Inference Time Comparison (s)}
    \label{enc_dec_time2}
\end{table}

Moreover, we compare the total runtime of a complete encryption-decryption cycle against several baselines: non-encrypted direct reconstruction (using both standard DDIM and O-BELM) and representative privacy-preserving diffusion model approaches. The results are summarized in Table~\ref{enc_dec_time2}.

It shows that our scheme incurs higher time overhead than standard DDIM and O-BELM. The main bottleneck is the generation of the key noise $\delta_t$; as shown in~\cite{wang2024belm}, the time cost of O‑BELM itself is similar to that of DDIM. However, compared to privacy‑preserving diffusion models that rely on homomorphic encryption or secure multi‑party computation, our approach offers a significant speedup. Those cryptographic primitives introduce high computational and communication overhead, often making them impractical for real‑time applications.

\section{Discussion \& Implementation Details}
\subsection{Avoiding Side-Channel Analysis}
The precise inversion in O-BELM relies on two consecutive latents, i.e., $\mathbf{x}_{i-1}$ and $\mathbf{x}_i$, which could open an opportunity for side-channel analysis aimed at inferring information from their relationship. Inspired by classical block-cipher designs, we introduce a whitening step, thereby effectively hardening the scheme against such side-channel attacks.

The purpose of whitening is to transform $\mathbf{x}_{i-1}$ using a simple method so that its distribution matches that of $\mathbf{x}_i$. An existing approach is to use the PRNG to generate noise $\delta_W$ satisfying:
\begin{equation}
    \delta_W\sim\mathcal{N}(\mu_{\mathbf{x}_i} - \mu_{\mathbf{x}_{i-1}}, \sigma^2_{\mathbf{x}_i} - \sigma^2_{\mathbf{x}_{i-1}})
\end{equation}
such that the addition of $\delta_W$ to $\mathbf{x}_{i-1}$ results in the same mean and variance as $\mathbf{x}_i$ while ensuring sufficient masking of the noise. We also propose an additional experiment to verify that the implementation of whitening does not significantly degrade the reconstruction or decrypted image, as is displayed in Tab.~\ref{tab:whitening} and Fig.~\ref{fig:whitening}. In the additional experiment, we set $\sigma=1.0$ and sampling steps $T=20$ as a typical setting for application.

\begin{table}
    \centering
    \begin{tabular}{lccc}
    \toprule
    Whitening & $\text{PSNR/dB}$ & $\text{SSIM}$ & $\text{MSE}$\\
    \midrule
    with whitening & 27.52 & 0.8321 & $2.8\times 10^{-3}$\\
    w/o whitening & 27.82 & 0.8620 & $2.5 \times 10^{-3}$\\
    \bottomrule
    \end{tabular}
    \caption{Comparison experiment of whitening}
    \label{tab:whitening}
\end{table}

\begin{figure}
    \centering
    \begin{subfigure}{0.49\linewidth}
        \includegraphics[width=\linewidth]{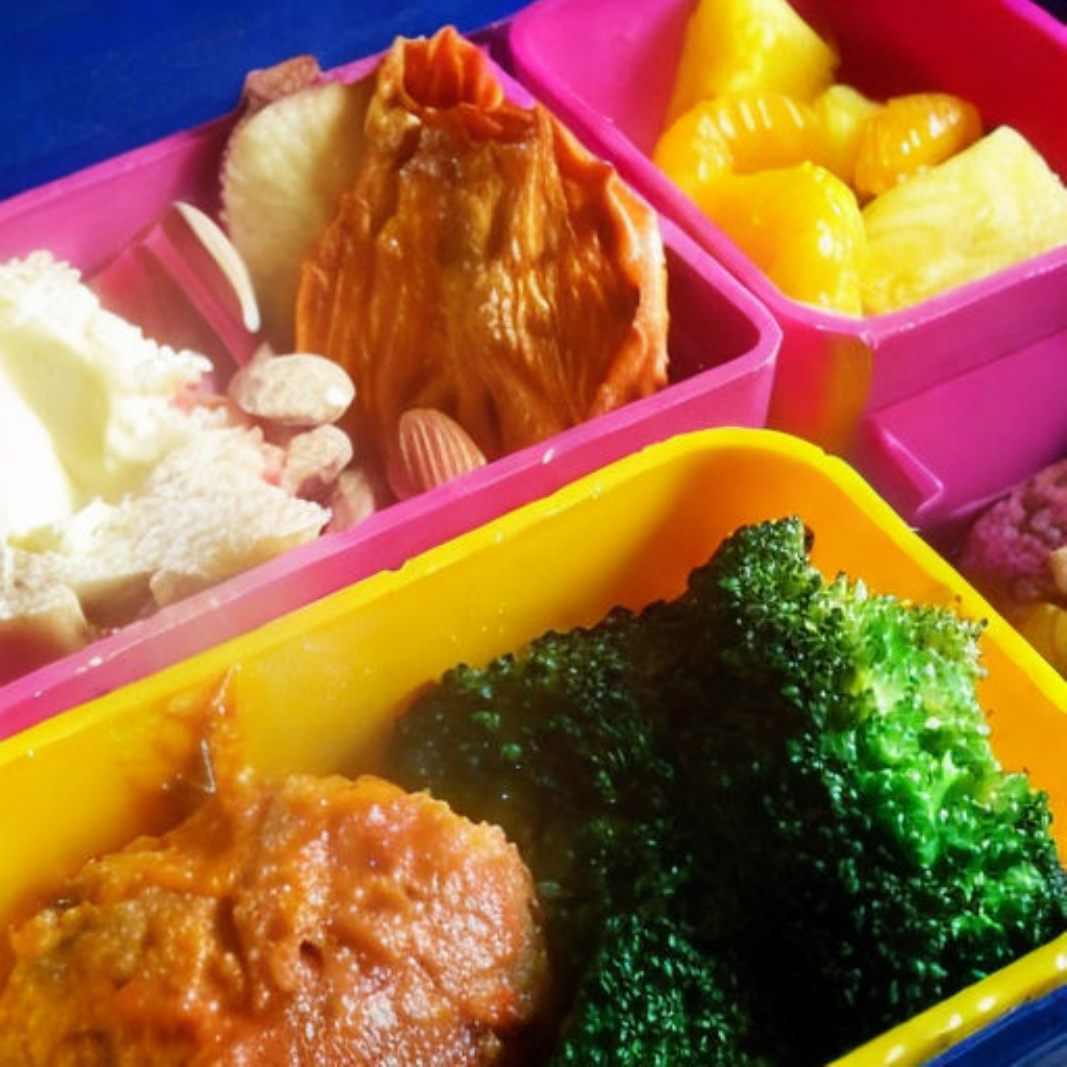}
        \caption{with whitening}
    \end{subfigure}
    \begin{subfigure}{0.49\linewidth}
        \includegraphics[width=\linewidth]{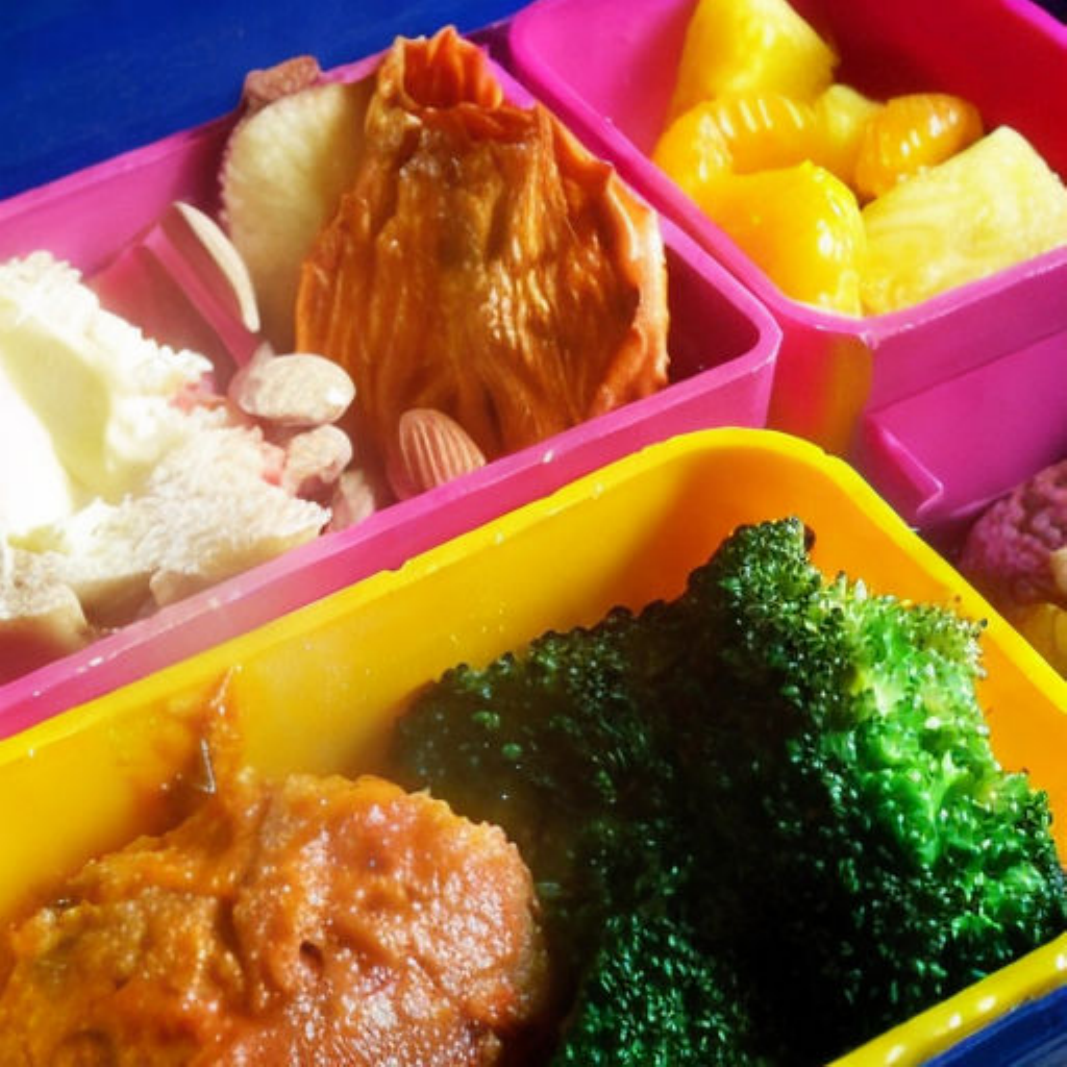}
        \caption{w/o whitening}
    \end{subfigure}
    \caption{Visual Quality of the whitening experiment}
    \label{fig:whitening}
\end{figure}

However, it is necessary to securely transfer the mean and variance of $\mathbf{x}_{i-1}$ for this process. Inspired by existing steganography algorithms, we derive a basic method as described in the Alg.~\ref{alg:whitening}, which applies perturbations to the whitened latent $\hat{\mathbf{x}}_{i-1}$ so that the encrypted mean and variance information can be embedded.

\begin{algorithm}[tb]
    \caption{Whitening with Secret Embedding}
    \label{alg:whitening}
    \textbf{Input:} latents $\mathbf{x}_{i-1},\mathbf{x}_i$, key $k$, step $i$, seed $n$ \\
    \textbf{Output:} whitened latent $\hat{\mathbf{x}}_{i-1}$
    \begin{algorithmic}[1]
        \STATE Compute mean and variance of $\mathbf{x}_{i-1},\mathbf{x}_i$.
        \STATE Generate $\delta_W \sim \mathcal{N}(\mu_{\mathbf{x}_i}-\mu_{\mathbf{x}_{i-1}},\ \sigma^2_{\mathbf{x}_i}-\sigma^2_{\mathbf{x}_{i-1}})$ using $\text{PRNG}(k,i,n)$.
        \STATE $\hat{\mathbf{x}}_{i-1} \leftarrow \mathbf{x}_{i-1} + \delta_W$.
        \STATE Generate a 128‑bit mask $B \leftarrow \text{PRNG}(k,i,n)$.
        \STATE Encrypt the 128‑bit parameter $M$ as $E = M \oplus B$, embedding $E$ into the first 128 entries of $\hat{\mathbf{x}}_{i-1}$.
        \STATE For each bit $B[j]$, if $B[j]=0$ then flip the sign of the $j$-th entry of $\hat{\mathbf{x}}_{i-1}$; otherwise leave it unchanged.
        \STATE \textbf{return} $\hat{\mathbf{x}}_{i-1}$
    \end{algorithmic}
\end{algorithm}

\subsection{Error Propagation on Sampling Process}
\label{err_prop_samp}
Similar to the security analysis, we examine the cross-model robustness from the perspective of Error Propagation Dynamics. This is because robustness across different models can also be modeled as error induced by the discrepancy in model-predicted noise. However, unlike the key-induced error, the model-prediction errors are introduced during the sampling/decryption process. Consequently, the behavior of model-prediction errors may differ from the analysis presented in Sec .~\ref {sec_proof}.

Our error propagation analysis again starts from Eqs.~\eqref{eq:key_inv} and~\eqref{eq:key_samp}, but here we modify the model prediction noise term $\varepsilon_\theta$. Let $\theta^{(E)}$ and $\theta^{(D)}$ denote the parameters of the encryption model and the decryption model, respectively. By substituting Eq.~\eqref{eq:key_inv} into Eq.~\eqref{eq:key_samp}, we can get the model-prediction error below:
\begin{equation}
    \label{model_predict_err}
    \mathbf{e}^{(D)}_{i-1}=a_i\mathbf{e}^{(D)}_{i+1} + b_i\mathbf{e}^{(D)}_{i} + c_i[\boldsymbol{\varepsilon}^{(E)}_{\theta}(\mathbf{x}_i,i)-\boldsymbol{\varepsilon}^{(D)}_{\theta}(\mathbf{x}_i^{(D)},i)]
\end{equation}
Consistent with Eq.~\eqref{error_prop}, we denote $\mathbf{e}^{(D)}_{i}$ as the error at step $i$, i.e., $\mathbf{e}^{(D)}_{i} = \mathbf{x}_i^{(D)} - \mathbf{x}_i$, where $\mathbf{x}_i^{(D)}$ is the latent at step $i$ utilizing the decryption model $\theta^{(D)}$. Without losing generality, we add a "mixed prediction noise" for linear approximation, i.e., Eq.~\eqref{mix_pre}.
\begin{align}
    \label{mix_pre}
    &\boldsymbol{\varepsilon}^{(E)}_{\theta}(\mathbf{x}_i,i)-\boldsymbol{\varepsilon}^{(D)}_{\theta}(\mathbf{x}_i,i) + \boldsymbol{\varepsilon}^{(D)}_{\theta}(\mathbf{x}_i,i)-\boldsymbol{\varepsilon}^{(D)}_{\theta}(\mathbf{x}_i^{(D)},i)\\
    \label{lin_approx}
    &\boldsymbol{\varepsilon}^{(D)}_{\theta}(\mathbf{x}_i,i)-\boldsymbol{\varepsilon}^{(D)}_{\theta}(\mathbf{x}_i^{(D)},i) \approx \mathbf{E}_i^{(D)}(\mathbf{x}_i^{(D)} - \mathbf{x}_i)=\mathbf{E}_i^{(D)}\mathbf{e}^{(D)}_{i}
\end{align}
Thus, similar to Eq.~\eqref{lin_error_prop}, we can obtain a second-order linear system through a linear approximation of the decryption model (Eq.~\eqref{lin_approx}). However, the driving term is now replaced by $\hat{\boldsymbol{\varepsilon}}_{ED}^{(i)}$ given in Eq.~\eqref{model_err}. This term represents the difference between the noises predicted by the encryption model and the decryption model for the same latent.

\begin{align}
\label{lin_model_error_prop}
    \mathbf{e}^{(D)}_{i-1}&=a_i\mathbf{e}^{(D)}_{i+1} + (b_i+c_i\mathbf{E}_i^{(D)})\mathbf{e}^{(D)}_{i} + c_i\hat{\boldsymbol{\varepsilon}}_{ED}^{(i)}\\
\label{model_err}
\hat{\boldsymbol{\varepsilon}}_{ED}^{(i)} &= \boldsymbol{\varepsilon}^{(E)}_{\theta}(\mathbf{x}_i,i)-\boldsymbol{\varepsilon}^{(D)}_{\theta}(\mathbf{x}_i,i)
\end{align}

Consequently, the error introduced by model discrepancy can be analyzed theoretically using a method similar to Sec.~\ref{sec_proof}. However, it should be noted that the final error expression is $\mathbf{e}_0^{(D)}$, since the error is introduced during the decryption/sampling process. 

\subsection{From Zero-stability to Cross-Model Robustness}
\label{cross_model_analysis}
Since the reversibility between encryption and decryption, it is intuitively expected that the impact of noise error decreases during the decryption process under Asm.~\ref{asm:linear_growth}. This implies that if the initial noise is identical, the disturbance introduced by the sampling process is often bounded. Notably, this intuition is formally supported by the properties of O-BELM below.

\begin{definition}[Zero-stability]
\label{def:zero_stability}
    A system described in Eq.~\eqref{lin_model_error_prop} is said to be zero-stable if, when its driving term is zero, i.e., $c_i=0$, there exists a constant $K$ such that for any two sequences $\{\mathbf{x}_i\}_{i=0}^N$ and $\{\mathbf{y}_i\}_{i=0}^N$ generated by the same sampling scheme but with different initial values, the following inequality holds:
    \begin{equation}
        \label{zero-stable}
        \|\mathbf{x}_i-\mathbf{y}_i\|\leq K \max_{i<j\leq N} \|\mathbf{x}_j-\mathbf{y}_j\|
    \end{equation}
\end{definition}

\begin{theorem}\cite{wang2024belm}
\label{zero_stab}
    O-BELM is zero-stable.
\end{theorem}

In Sec.~\ref{err_prop_samp}, we have shown that the error introduced by the model discrepancy can be analyzed theoretically using a similar approach as in Sec.~\ref{sec_proof}. Likewise, the total error can also be written as the sum of noise injected at each step with a list of specific coefficients $\mathbf{\Psi}_t^{(D)}$ as Eq.~\eqref{total_error}, which is shown below:
\begin{equation}
    \label{total_err_samp}
    \mathbf{e}_0^{(D)} = \sum_{t=0}^{T-1}\mathbf{\Psi}_t^{(D)}\hat{\boldsymbol{\varepsilon}}_{ED}^{(t)}
\end{equation}

To simplify the analysis, we assume that the effect of this noise becomes negligible in the subsequent predictions of the diffusion model, i.e., there exists a negligible $\varepsilon >0$ such that $\|c_i\mathbf{E}_i^{(D)}\|<\varepsilon$. Thus, we obtain the following proposition:
\begin{proposition}
\label{robust_bound}
    Under the above assumption, the propagation of the model-prediction error noise remains bounded.
\end{proposition}
\begin{proof}
Under this assumption, each term in the summation of Eq.~\eqref{total_err_samp} can be regarded as the scenario where noise is introduced only at step $i$. Consequently, the effect of the error can be viewed as the combined contribution of $N$ independent sampling processes with sampling lengths of $1,2,\cdots, N$ steps, respectively. Each sampling process is initialized with an error (Eq.~\eqref{samp_err}).
\begin{equation}
    \label{samp_err}
    \|\mathbf{x}_i-\mathbf{y}_i\|=\|c_i\hat{\boldsymbol{\varepsilon}}_{ED}^{(i)}\|
\end{equation}

According to the zero-stability of O-BELM (Thm.~\ref{zero_stab}), we can infer that there exists a constant $K_i$ to bound the contribution of i-th model-prediction error, i.e.
\begin{equation}
    \label{err_bnd_samp}
    \|\mathbf{\Psi}_i^{(D)}\hat{\boldsymbol{\varepsilon}}_{ED}^{(i)}\| \leq c_iK_i\|\hat{\boldsymbol{\varepsilon}}_{ED}^{(i)}\|
\end{equation}

Substituting Eq.~\eqref{err_bnd_samp} into Eq.~\eqref{total_err_samp} will result in the goal of the proof, that is:
\begin{equation}
    \|\mathbf{e}_0^{(D)}\|\leq \sum_{t=0}^{T-1} \|\mathbf{\Psi}_t^{(D)}\hat{\boldsymbol{\varepsilon}}_{ED}^{(t)}\|\leq \sum_{t=0}^{T-1} c_tK_t\|\hat{\boldsymbol{\varepsilon}}_{ED}^{(t)}\|
\end{equation}
which indicates a specific upper bound.
\end{proof}

Based on Prop.~\ref{robust_bound}, we conclude that the cross-model robustness essentially stems from the boundedness of the model-prediction error. This analysis is supported by the experimental observation in Sec.~\ref{cross_model_robust} that the reconstruction quality with the key is positively correlated with the quality without the key.

\subsection{Practical Limitations: Chosen-Ciphertext Security}
Our cross‑model robustness analysis reveals a fundamental property of the sampling process: it satisfies zero‑stability (Def.~\ref{def:zero_stability} and Thm.~\ref{zero_stab}). For the O-BELM sampling, the distance between two decryption paths should satisfy Eq.~\eqref{zero-stable}. In the context of decryption, if the ciphertext $\mathbf{x}_T^{(\delta)}$ is perturbed, the resulting error at any intermediate step $\mathbf{x}_i^{(\delta)}$ cannot exceed a constant multiple of the initial perturbation. Consequently, our scheme does not aim to provide chosen-ciphertext (CCA) security, and intermediate latents should be treated as sensitive.

\textbf{Practical attack: partial decryption.}
An adversary who gains access to $\mathbf{x}_i^{(\delta)}$ (e.g., through a side channel or memory dump) can treat it as a new ciphertext with fewer remaining steps. Since the noise has already been partially applied, the remaining steps may not provide sufficient amplification to hide the key mismatch. A concrete countermeasure is to never expose intermediate latents to untrusted parties and to run the entire decryption in a trusted environment.

\textbf{Why IND‑CPA remains unaffected.}
The IND‑CPA security of our scheme relies on the difficulty of distinguishing two ciphertexts corresponding to two different plaintexts. From the perspective of a chosen‑plaintext adversary, the only available information is the final ciphertext $\mathbf{x}_T^{(\delta)}$. Constructing a pair of plaintexts whose ciphertexts are statistically close is computationally hard because the forward (encryption) direction exhibits exponential error amplification. The zero‑stability property does not help an adversary in the CPA setting because the adversary cannot query the decryption oracle or obtain intermediate states. Hence, the IND‑CPA guarantee stands independent of the CCA weakness.

\section{Related Works}
\subsection{Securing Generative Content}
\label{Sec_Gen}

Current image protection methods for diffusion models can be broadly classified into three categories: post-hoc tracing, proactive defenses, and privacy‑preserving computation.

\textbf{Post‑hoc tracing.}
Techniques such as image watermarking~\cite{zhu2018hidden} embed invisible marks into generated content for ownership verification. Model watermarking~\cite{fernandez2023stable} fingerprints the diffusion model itself. These methods are reactive: they detect misuse after the fact but cannot prevent unauthorized reconstruction or inversion.

\textbf{Proactive defenses.}
Adversarial perturbation methods (e.g., EditShield~\cite{chen2024editshield}) disrupt specific attacks but do not offer a generic key‑controlled encryption mechanism. Other access‑control frameworks~\cite{gai2025pcdiff, Takana2025access, lei2025secure} rely on fine‑tuning the decoder or modifying the sampling process, yet they lack formal security analysis.

\textbf{Privacy‑preserving diffusion.}
A separate line of work uses homomorphic encryption (HE)~\cite{chen2024privacy, he2025private} or secure MPC~\cite{zhao2024cipherdm} to protect text prompts or model weights. These approaches assume a black‑box setting, incur huge overhead, and are designed for text‑to‑image generation. More importantly, the non‑linear operations in the U‑Net (activation functions) make HE‑based image‑to‑image (I2I) reconstruction infeasible.

\textbf{Existing encryption methods.}
Some works encrypt only the final latent $\mathbf{x}_T$~\cite{kaur2020comprehensive}, leaving the entire denoising (or inversion) pipeline unprotected. Others use DNN‑based encryption~\cite{guo2024image, huang2025new} without rigorous security reductions. None of the above addresses the white‑box scenario where model parameters are fully public, and an adversary can simply download the model and invert any obtained latent.

\subsection{Our Positioning}
\label{Our_Pos}

To overcome the limitations above, we propose a framework that embeds cryptographic control directly into the diffusion inversion process. Compared to existing methods, our framework has the following features:

\begin{itemize}
    \item \textbf{White‑box setting.} It assumes the diffusion model parameters are public; security relies only on the secret key, not on hiding the model.
    \item \textbf{Exact reversibility.} Based on the strictly invertible O‑BELM sampler (Sec.~\ref{Inv_Re}), we inject key‑dependent noise $\delta_t$ into the noise prediction term. The correct key cancels the noise exactly; an incorrect key leads to exponential amplification of error.
    \item \textbf{Provable IND‑CPA security.} In Sec.~\ref{sec_proof}, we derive an upper bound on the adversary's distinguishing advantage, which is exponentially small in the number of steps.
    \item \textbf{Lightweight computation.} Compared to HE/MPC‑based solutions, our method only adds PRNG operations and a few operations per step. Runtime is comparable to standard inversion (O‑BELM).
    \item \textbf{Designed for image‑to‑image reconstruction.} The scheme targets reconstructing the original image from a latent, serving as a building block for downstream applications such as editing.
\end{itemize}
\section{Conclusion}
We presented a key-controlled inversion framework for diffusion models that enables conditional image reconstruction without hiding model parameters. By injecting key-dependent noise into the noise prediction term of a strictly invertible sampler, our scheme guarantees exact reconstruction with the correct key and exponential error amplification otherwise. 

We analyzed its security by reducing IND-CPA attacks to a practical hypothesis test and deriving an advantage bound that depends on the accumulated error variance $\mathbf{\Sigma}_T$, thereby ensuring the existence of a security parameter $\lambda$. We also empirically validate it through experiments and further discuss cross-model robustness, where decryption quality depends primarily on the inter-model differences, and key-noise injection does not substantially amplify the performance drop inherent to cross-model reconstruction.

Our results demonstrate that error propagation in diffusion models can be turned into a powerful security asset. Future work may explore access-conditioned image editing and leverage its cross-model robustness to develop data hiding in the encrypted domain~\cite{ke2021reversible} and watermarking with client-side embedding~\cite{xiao2025preview}, further expanding its application scenarios.

\newpage

\IEEEtriggercmd{\enlargethispage{-5in}}


\bibliographystyle{IEEEtran}
\bibliography{ref}

@inproceedings{sohl2015deep,
  title={Deep unsupervised learning using nonequilibrium thermodynamics},
  author={Sohl-Dickstein, Jascha and Weiss, Eric and Maheswaranathan, Niru and Ganguli, Surya},
  booktitle={International conference on machine learning},
  pages={2256--2265},
  year={2015},
  organization={pmlr}
}

@article{jiang2024survey,
  title={A survey of multimodal controllable diffusion models},
  author={Jiang, Rui and Zheng, Guang-Cong and Li, Teng and Yang, Tian-Rui and Wang, Jing-Dong and Li, Xi},
  journal={Journal of Computer Science and Technology},
  volume={39},
  number={3},
  pages={509--541},
  year={2024},
  publisher={Springer}
}

@article{ho2020denoising,
  title={Denoising diffusion probabilistic models},
  author={Ho, Jonathan and Jain, Ajay and Abbeel, Pieter},
  journal={Advances in neural information processing systems},
  volume={33},
  pages={6840--6851},
  year={2020}
}

@article{song2020denoising,
  title={Denoising diffusion implicit models},
  author={Song, Jiaming and Meng, Chenlin and Ermon, Stefano},
  journal={arXiv preprint arXiv:2010.02502},
  year={2020}
}

@inproceedings{mokady2023null,
  title={Null-text inversion for editing real images using guided diffusion models},
  author={Mokady, Ron and Hertz, Amir and Aberman, Kfir and Pritch, Yael and Cohen-Or, Daniel},
  booktitle={Proceedings of the IEEE/CVF conference on computer vision and pattern recognition},
  pages={6038--6047},
  year={2023}
}

@inproceedings{wallace2023edict,
  title={Edict: Exact diffusion inversion via coupled transformations},
  author={Wallace, Bram and Gokul, Akash and Naik, Nikhil},
  booktitle={Proceedings of the IEEE/CVF Conference on Computer Vision and Pattern Recognition},
  pages={22532--22541},
  year={2023}
}

@inproceedings{zhang2024exact,
  title={Exact diffusion inversion via bidirectional integration approximation},
  author={Zhang, Guoqiang and Lewis, Jonathan P and Kleijn, W Bastiaan},
  booktitle={European Conference on Computer Vision},
  pages={19--36},
  year={2024},
  organization={Springer}
}

@article{wang2024belm,
  title={Belm: Bidirectional explicit linear multi-step sampler for exact inversion in diffusion models},
  author={Wang, Fangyikang and Yin, Hubery and Dong, Yue-Jiang and Zhu, Huminhao and Zhao, Hanbin and Qian, Hui and Li, Chen and others},
  journal={Advances in Neural Information Processing Systems},
  volume={37},
  pages={46118--46159},
  year={2024}
}

@article{gretton2012kernel,
  title={A kernel two-sample test},
  author={Gretton, Arthur and Borgwardt, Karsten M and Rasch, Malte J and Sch{\"o}lkopf, Bernhard and Smola, Alexander},
  journal={The journal of machine learning research},
  volume={13},
  number={1},
  pages={723--773},
  year={2012},
  publisher={JMLR. org}
}

@article{huang2025diffusion,
  title={Diffusion model-based image editing: A survey},
  author={Huang, Yi and Huang, Jiancheng and Liu, Yifan and Yan, Mingfu and Lv, Jiaxi and Liu, Jianzhuang and Xiong, Wei and Zhang, He and Cao, Liangliang and Chen, Shifeng},
  journal={IEEE Transactions on Pattern Analysis and Machine Intelligence},
  year={2025},
  publisher={IEEE}
}

@inproceedings{zhu2018hidden,
  title={Hidden: Hiding data with deep networks},
  author={Zhu, Jiren and Kaplan, Russell and Johnson, Justin and Fei-Fei, Li},
  booktitle={Proceedings of the European conference on computer vision (ECCV)},
  pages={657--672},
  year={2018}
}

@inproceedings{fernandez2023stable,
  title={The stable signature: Rooting watermarks in latent diffusion models},
  author={Fernandez, Pierre and Couairon, Guillaume and J{\'e}gou, Herv{\'e} and Douze, Matthijs and Furon, Teddy},
  booktitle={Proceedings of the IEEE/CVF International Conference on Computer Vision},
  pages={22466--22477},
  year={2023}
}

@inproceedings{chen2024editshield,
  title={Editshield: Protecting unauthorized image editing by instruction-guided diffusion models},
  author={Chen, Ruoxi and Jin, Haibo and Liu, Yixin and Chen, Jinyin and Wang, Haohan and Sun, Lichao},
  booktitle={European Conference on Computer Vision},
  pages={126--142},
  year={2024},
  organization={Springer}
}

@inproceedings{liu2024stable,
  title={Stable unlearnable example: Enhancing the robustness of unlearnable examples via stable error-minimizing noise},
  author={Liu, Yixin and Xu, Kaidi and Chen, Xun and Sun, Lichao},
  booktitle={Proceedings of the AAAI Conference on Artificial Intelligence},
  volume={38},
  number={4},
  pages={3783--3791},
  year={2024}
}

@article{he2025private,
  title={Private Sampling of Latent Diffusion Models for Encrypted Prompt},
  author={He, Guanghui and Ren, Yanli and Cai, Xiaoqiu and Feng, Guorui and Zhang, Xinpeng},
  journal={IEEE Transactions on Circuits and Systems for Video Technology},
  year={2025},
  publisher={IEEE}
}

@inproceedings{zhao2024cipherdm,
  title={Cipherdm: Secure three-party inference for diffusion model sampling},
  author={Zhao, Xin and Chen, Xiaojun and Chen, Xudong and Li, He and Fan, Tingyu and Zhao, Zhendong},
  booktitle={European Conference on Computer Vision},
  pages={288--305},
  year={2024},
  organization={Springer}
}

@article{kaur2020comprehensive,
  title={A comprehensive review on image encryption techniques.},
  author={Kaur, Manjit and Kumar, Vijay},
  journal={Archives of Computational Methods in Engineering},
  volume={27},
  number={1},
  year={2020}
}

@article{chen2024privacy,
  title={Privacy-preserving diffusion model using homomorphic encryption},
  author={Chen, Yaojian and Yan, Qiben},
  journal={arXiv preprint arXiv:2403.05794},
  year={2024}
}

@article{he2025private1,
  title={Private image synthesis of latent diffusion model with the ciphertext of prompt},
  author={He, Guanghui and Ren, Yanli and Li, Gaojian and Feng, Guorui and Zhang, Xinpeng},
  journal={Neural Networks},
  pages={107678},
  year={2025},
  publisher={Elsevier}
}

@inproceedings{guo2024image,
  title={Image Encryption and Compression Based on Reversed Diffusion Model},
  author={Guo, Yilin and Chang, Jianhui and Zhang, Yuhuai and Zhang, Jian and Ma, Siwei},
  booktitle={2024 Picture Coding Symposium (PCS)},
  pages={1--5},
  year={2024},
  organization={IEEE}
}

@article{huang2025new,
  title={New Framework of Robust Image Encryption},
  author={Huang, Lin and Qin, Chuan and Feng, Guorui and Luo, Xiangyang and Zhang, Xinpeng},
  journal={ACM Transactions on Multimedia Computing, Communications and Applications},
  volume={21},
  number={3},
  pages={1--22},
  year={2025},
  publisher={ACM New York, NY}
}

@article{gai2025pcdiff,
  title={PCDiff: Proactive Control for Ownership Protection in Diffusion Models with Watermark Compatibility},
  author={Gai, Keke and Shen, Ziyue and Yu, Jing and Zhu, Liehuang and Wu, Qi},
  journal={arXiv preprint arXiv:2504.11774},
  year={2025}
}

@inproceedings{Takana2025access,
  title={Access Control for Diffusion Models by Random Masking the Covariance of Initial Noise Distribution},
  author={Temma Tanaka and Kazuaki Nakamura},
  booktitle={Asia Pacific Signal and Information Processing Association Annual Summit and Conference},
  pages={2062--2067},
  year={2025},
  organization={IEEE}
}

@inproceedings{deng2009imagenet,
  title={Imagenet: A large-scale hierarchical image database},
  author={Deng, Jia and Dong, Wei and Socher, Richard and Li, Li-Jia and Li, Kai and Fei-Fei, Li},
  booktitle={2009 IEEE conference on computer vision and pattern recognition},
  pages={248--255},
  year={2009},
  organization={Ieee}
}

@article{karras2017progressive,
  title={Progressive growing of gans for improved quality, stability, and variation},
  author={Karras, Tero and Aila, Timo and Laine, Samuli and Lehtinen, Jaakko},
  journal={arXiv preprint arXiv:1710.10196},
  year={2017}
}

@inproceedings{lin2014microsoft,
  title={Microsoft coco: Common objects in context},
  author={Lin, Tsung-Yi and Maire, Michael and Belongie, Serge and Hays, James and Perona, Pietro and Ramanan, Deva and Doll{\'a}r, Piotr and Zitnick, C Lawrence},
  booktitle={European conference on computer vision},
  pages={740--755},
  year={2014},
  organization={Springer}
}

@article{rombach2021high,
  title={High-Resolution Image Synthesis with Latent Diffusion Models},
  author={Robin Rombach and A. Blattmann and Dominik Lorenz and Patrick Esser and Bj{\"o}rn Ommer},
  journal={2022 IEEE/CVF Conference on Computer Vision and Pattern Recognition (CVPR)},
  year={2021},
  pages={10674-10685},
}

@inproceedings{lei2025secure,
  title={Secure and Efficient Watermarking for Latent Diffusion Models in Model Distribution Scenarios},
  author={Li Lei and Keke Gai and Jing Yu and Liehuang Zhu and Qi Wu},
  booktitle={International Joint Conference on Artificial Intelligence},
  year={2025},
}

@inproceedings{nichol2021improved,
  title={Improved denoising diffusion probabilistic models},
  author={Nichol, Alexander Quinn and Dhariwal, Prafulla},
  booktitle={International conference on machine learning},
  pages={8162--8171},
  year={2021},
  organization={PMLR}
}

@article{heusel2017gans,
  title={Gans trained by a two time-scale update rule converge to a local nash equilibrium},
  author={Heusel, Martin and Ramsauer, Hubert and Unterthiner, Thomas and Nessler, Bernhard and Hochreiter, Sepp},
  journal={Advances in neural information processing systems},
  volume={30},
  year={2017}
}

@article{xiao2025preview,
  title={Preview Helps Selection: Previewable Image Watermarking with Client-Side Embedding},
  author={Xiao, Xiangli and Zhang, Yushu and Hua, Zhongyun and Xia, Zhihua and Weng, Jian},
  journal={IEEE Transactions on Dependable and Secure Computing},
  year={2025},
  publisher={IEEE}
}

@article{ke2021reversible,
  title={A reversible data hiding scheme in encrypted domain for secret image sharing based on Chinese remainder theorem},
  author={Ke, Yan and Zhang, Minqing and Zhang, Xinpeng and Liu, Jia and Su, Tingting and Yang, Xiaoyuan},
  journal={IEEE Transactions on Circuits and Systems for Video Technology},
  volume={32},
  number={4},
  pages={2469--2481},
  year={2021},
  publisher={IEEE}
}

@article{li2023error,
  title={On error propagation of diffusion models},
  author={Li, Yangming and van der Schaar, Mihaela},
  journal={arXiv preprint arXiv:2308.05021},
  year={2023}
}
%

\end{document}